\documentclass[journal,final,twocolumn]{IEEEtran} 
\usepackage{amsmath,amssymb,amsfonts}
\usepackage{algorithmic}
\usepackage{graphicx}
\usepackage{textcomp}
\usepackage{subcaption}
\usepackage{url}
\usepackage{xcolor}
\usepackage{overpic}

\usepackage{bbm}
\usepackage{amssymb}
\usepackage{amsmath}
\usepackage{amsthm}
\usepackage{aliascnt}
\usepackage{mathtools}

\usepackage{xr-hyper}
\definecolor{linkcolor}{HTML}{005075}
\definecolor{ForestGreen}{HTML}{009B55}
\usepackage[colorlinks = true, linkcolor = linkcolor, urlcolor=linkcolor, citecolor = linkcolor, filecolor=linkcolor]{hyperref}

\theoremstyle{plain}
\newtheorem{theorem}{Theorem}[section]

\newaliascnt{proposition}{theorem}
\newtheorem{proposition}[proposition]{Proposition}
\aliascntresetthe{proposition}

\newaliascnt{corollary}{theorem}
\newtheorem{corollary}[corollary]{Corollary}
\aliascntresetthe{corollary}

\newaliascnt{lemma}{theorem}
\newtheorem{lemma}[lemma]{Lemma}
\aliascntresetthe{lemma}

\theoremstyle{definition}
\newtheorem{definition}{Definition}[section]

\theoremstyle{remark}
\newtheorem{rmk}{Remark}[section]

\newtheorem*{rmk*}{Remark}
\newtheorem*{fact*}{Fact}

\newenvironment{remark}
{\pushQED{\qed}\rmk}
{\popQED\endrmk}

\newcommand{\ie}{\emph{i.\@e.\@,~}}
\newcommand{\eg}{\emph{e.\@g.\@,~}}
\newcommand{\etc}{\emph{\textit{etc.\@}}}

\newcommand{\reals}{\ensuremath{\mathbb{R}}}
\newcommand{\nats}{\ensuremath{\mathbb{N}}}

\newcommand{\complex}{\ensuremath{\mathbb{C}}}
\newcommand{\sph}{\ensuremath{\mathbb{S}}}
\newcommand{\sphn}[1][n]{\ensuremath{\mathbb{S}^{{#1}-1}}}
\newcommand{\ip}[1]{\ensuremath{\langle #1\rangle}}

\newcommand{\iid}{\ensuremath{\overset{\text{iid}}{\sim}}}
\renewcommand{\i}{\mathrm{i}}
\newcommand{\DistBer}{\mathrm{Ber}}
\newcommand{\DistBinom}{\mathrm{Binom}}

\newcommand{\bmat}[2]{\ensuremath{\left[
      \begin{array}{#1}
        #2
      \end{array}
    \right]}}

\DeclareMathOperator*{\pr}{\mathbb P}

\DeclareMathOperator{\diag}{\mathrm{diag}}

\DeclareMathOperator{\range}{\mathcal{R}}
\DeclareMathOperator{\rre}{rre}
\DeclareMathOperator{\sigmoid}{\mathrm{s}}

\DeclareMathOperator{\grassmannian}{\Gamma}

\DeclareMathOperator{\Span}{\mathrm{span}}

\DeclareMathOperator{\proj}{\Pi}

\DeclareMathOperator{\relu}{ReLU}

\DeclareMathOperator{\E}{\mathbb{E}}

\renewcommand{\Re}{\mathop{\mathrm{Re}}}
\renewcommand{\Im}{\mathop{\mathrm{Im}}}

\newcommand{\1}{\mathbbm{1}}

\DeclarePairedDelimiterX{\kldivx}[2]{(}{)}{%
  #1\;\delimsize\|\;#2%
}
\newcommand{\kldiv}{\operatorname{D_{\mathrm{KL}}}\kldivx}

\theoremstyle{remark}
\newtheorem*{myexmp}{Example}
\newenvironment{myexample}
{\pushQED{\qed}\myexmp}
{\popQED\endmyexmp}

\makeatletter
\newcommand*{\addFileDependency}[1]{
  \typeout{(#1)}
  \@addtofilelist{#1}
  \IfFileExists{#1}{}{\typeout{No file #1.}}
}
\makeatother

\theoremstyle{plain}
\newtheorem*{sketch}{Theorem Sketch}
\newtheorem*{problem}{Open problem}

\title{A coherence parameter characterizing generative compressed sensing with Fourier measurements}

\author{Aaron Berk, Simone Brugiapaglia, Babhru Joshi, Yaniv Plan, Matthew Scott and \"Ozg\"ur Yilmaz
\thanks{This manuscript was submitted 19 July, 2022. A.~Berk is partially supported by Institut de valorisation des don\'ees (IVADO) and Centre de recherches en math\'ematiques (CRM) Applied Math Lab. S.~Brugiapaglia is partially supported by NSERC grant RGPIN-2020-06766 and the Faculty of Arts and Science of Concordia University. B.~Joshi is supported in part by the Pacific Institute for the Mathematical Sciences (PIMS). Y.~Plan is partially supported by an NSERC Discovery Grant (GR009284), an NSERC Discovery Accelerator Supplement (GR007657), and a Tier II Canada Research Chair in Data Science (GR009243). \"O.~Yilmaz is supported in part by an NSERC Discovery Grant (22R82411), and a UBC Data Science Institute grant.}
\thanks{A.~Berk is with McGill University, Montr\'eal, QC, Canada (aaron.berk@mcgill.ca)}
\thanks{S.~Brugiapaglia is with Concordia University, Montr\'eal, QC, Canada\newline (simone.brugiapaglia@concordia.ca)}
\thanks{B.~Joshi, Y.~Plan, M.~Scott \& \"O.~Yilmaz are with the University of British Columbia, Vancouver, BC, Canada ($\{$b.joshi, yaniv, matthewscott, oyilmaz$\}$@math.ubc.ca)}}

\begin{document}

\maketitle


\begin{abstract}
  In~\cite{bora2017compressed}, a mathematical framework was developed for
  compressed sensing guarantees in the setting where the measurement matrix is
  Gaussian and the signal structure is the range of a generative neural network
  (GNN). The problem of compressed sensing with GNNs has since been extensively
  analyzed when the measurement matrix and/or network weights follow a
  subgaussian distribution. We move beyond the subgaussian assumption, to
  measurement matrices that are derived by sampling uniformly at random rows of
  a unitary matrix (including subsampled Fourier measurements as a special
  case). Specifically, we prove the first known restricted isometry guarantee
  for generative compressed sensing (GCS) with \emph{subsampled isometries} and
  provide recovery bounds, addressing an open problem
  of~\cite[p.~10]{scarlett2022theoretical}. Recovery efficacy is characterized
  by the \emph{coherence}, a new parameter, which measures the interplay between
  the range of the network and the measurement matrix. Our approach relies on
  subspace counting arguments and ideas central to high-dimensional
  probability. Furthermore, we propose a regularization strategy for training
  GNNs to have favourable coherence with the measurement operator. We provide
  compelling numerical simulations that support this regularized training
  strategy: our strategy yields low coherence networks that require fewer
  measurements for signal recovery. This, together with our theoretical results,
  supports coherence as a natural quantity for characterizing GCS with
  subsampled isometries.
\end{abstract}


\section{Introduction}

The solution of underdetermined linear inverse problems has many important
applications including geophysics~\cite{kumar2015source, herrmann2012fighting}
and medical imaging~\cite{adcock2021compressive, lustig2008compressed}. In
particular, compressed sensing permits accurate and stable recovery of signals
that are well represented by one of a certain set of structural proxies (\eg
sparsity)~\cite{adcock2021compressive, foucart2017mathematical}. Moreover, this
recovery is effected using an order-optimal number of random
measurements~\cite{foucart2017mathematical}. In applications like medical
imaging~\cite{adcock2021compressive}, the measurement matrices under
consideration are derived from a bounded orthonormal system (a unitary matrix
with bounded entries), which complicates the theoretical analysis. Furthermore,
for such applications one desires a highly effective representation for encoding
the images. Developing a theoretical analysis that properly accounts for
realistic measurement paradigms and complexly designed image representations is
nontrivial in general~\cite{bora2017compressed, foucart2017mathematical,
  jalal2021robust}. For example, there has been much work validating that
\emph{generative neural networks} (GNNs) are highly effective at representing
natural signals~\cite{kingma2013auto, goodfellow2014generative,
  radford2015unsupervised}. In this vein, recent work has shown promising
empirical results for compressed sensing with realistic measurement matrices
when the structural proxy is a {GNN}. Other recent work has established recovery
guarantees for compressed sensing when the structural proxy is a {GNN} and the
measurement matrix is subgaussian~\cite{bora2017compressed}.
(See~\autoref{sec:related-work} for a fuller depiction of related aspects of
this problem.) However, an open problem is the
following~\cite[p.~10]{scarlett2022theoretical}:

\begin{problem}[subIso GCS]
  \label{problem:subIso-GCS}
  A theoretical analysis of compressed sensing when the measurement matrix is
  structured (\eg a randomly subsampled unitary matrix) and the signal model
  proxy is a GNN.
\end{problem}

Broadly, we approach a solution to \nameref{problem:subIso-GCS} as follows. For
a matrix $A \in \complex^{m \times n}$, a particular GNN \emph{architecture}
$G : \reals^{k} \to \reals^{n}$ and an unknown signal $x_{0} \in \range(G)$, the
range of $G$, we determine the conditions (on $A, G, x_{0}$, \etc) under which
it is possible to approximately recover $x_{0}$ from noisy linear measurements
$b = Ax_{0} + \eta$ by (approximately) solving an optimization problem of the
form
\begin{align}
  \label{eq:gnn-opt}
  \min_{z \in \reals^{k}} \|b - A G(z) \|_{2}. 
\end{align}
Above, $\eta \in \complex^{m}$ is some unknown corruption. Specifically, we are
interested in establishing sample complexity bounds (lower bounds on $m$) for
realistic measurement matrices $A$ --- where $A$ is an underdetermined matrix
randomly subsampled from a unitary matrix. Namely, the rows of $A$ have been
sampled uniformly at random without replacement from a unitary matrix
$U \in \complex^{n \times n}$. We next present a mathematical description of the
subsampling of the rows, described similarly
to~\cite[Section~4]{dirksen2015tail}.

\begin{definition}[subsampled isometry]
  \label{def:subsampled-isometry}
  Let $2 \leq m \leq n < \infty$ be integers and let
  $U \in \complex^{n \times n}$ be a unitary matrix. Let
  $\theta := (\theta_{i})_{i \in [n]}$ be an iid Bernoulli random vector:
  $\theta_{i} \iid \DistBer(m/n)$. Define the set
  $\mathcal{J} := \{ j : \theta_{j} = 1\}$ and enumerate the elements of
  $\mathcal{J}$ as $j_{1}, \ldots, j_{\tilde m}$ where
  $\tilde m := |\mathcal{J}|$ is a binomial random variable with
  $\E \tilde m = m$. Let $A \in \complex^{\tilde m \times n}$ be the matrix
  whose $i$th row is $\frac{\sqrt n}{\sqrt m}U_{j_{i}}, i \in [\tilde m]$. We
  call $A$ an $(m, U)$-subsampled isometry. When there is no risk of confusion
  we simply refer to $A$ as a subsampled isometry and implicitly acknowledge the
  existence of an $(m, \theta, U)$ giving rise to $A$.
\end{definition}

With $A$ so defined, $A$ is \emph{isotropic}:
$\E A^* A = \sum_{i =1}^{n} U_{i}^{*} U_{i} = I_{n}$ where $I_{n}$ is the
$n\times n$ identity matrix.

\begin{remark}
  \label{rmk:dft}
  An important example of a matrix isometry is the discrete orthogonal system
  given by the (discrete) Fourier basis. For example, ``in applications such as
  medical imaging, one is confined to using subsampled Fourier measurements due
  to the inherent design of the
  hardware''~\cite[p.~10]{scarlett2022theoretical}. Let
  $F = (F_{ij})_{i, j \in [n]}$ with
  $F_{ij} := \frac{1}{\sqrt n} \exp(2\pi \i (i-1)(j - 1)/n)$ for each
  $i, j \in [n]$. Here, we have used $\i$ to denote the complex number
  satisfying $\i^{2} = -1$. The matrix $F$ is known as the discrete Fourier
  transform (DFT) matrix, and has important roles in signal processing and
  numerical computation~\cite{adcock2021compressive, scarlett2022theoretical,
    jalal2021robust}. Thus, all of our results apply, in particular, when the
  measurement matrix is a subsampled DFT.
\end{remark}

Lastly, we introduce the kind of GNN to which we restrict our attention in this
work. Namely, we study $\relu$-activated expansive neural networks, where ReLU
is the so-called rectified linear unit defined as $\sigma(x) := \max(x, 0)$,
acting element-wise on the entries of $x$.

\begin{definition}[$(k, d, n)$-generative network]
  \label{def:generative-network}
  Fix the integers $2 \leq k := k_{0} \leq k_{1}, \ldots, k_{d}$ where
  $k_d := n < \infty$, and suppose for $i \in [d]$ that
  $W^{(i)} \in \reals^{k_{i}\times k_{i-1}}$. A $(k,d,n)$-generative network is
  a function $G : \reals^{k} \to \reals^{n}$ of the form
  \begin{align*}
    G(z) := W^{(d)} \sigma \left( \cdots W^{(2)} \sigma \left( W^{(1)} z \right) \right). 
  \end{align*}
\end{definition}

\begin{remark}
  \label{rmk:networks-with-biases}
  In practice, ReLU generative networks often use \emph{biases}, which are
  learned parameters in addition to the weight matrices. Such networks have the
  form
  \begin{align*}
    G_{\text{bias}}(z) %
    := W^{(d)} \sigma(\cdots \sigma(W^{(1)} z + b^{(1)})\cdots) + b^{(d)}, 
  \end{align*}
  $b^{(i)} \in \reals^{k_{i}}, i \in [d]$. Let
  $\mathcal{H} := \{ z \in \reals^{k+1} : z_{k+1} = 1\}$ and define the
  \textit{augmented matrices}
  $\tilde W^{(i)} := \bmat{cc}{W^{(i)} & b^{(i)}\\ \mathbf{0} & 1}$. In the last
  layer, remove the last row of the augmented matrix. Let
  $\tilde G : \reals^{k+1} \to \reals^{n}$ be the $(k+1,d,n)$-generative network
  having weight matrices $\tilde W^{(i)}$. Since
  $G_{\text{bias}} = \left.\tilde G\right|_{\mathcal{H}}$, we have
  $\range(G_\text{bias}) \subset \range(\tilde{G})$ (\ie any network with biases
  has range contained in a similar network without biases that has code
  dimension augmented by 1). Our theory applies directly to the network without
  biases, $\tilde{G}$, and due to the containment, all results given for
  $\tilde{G}$ extend to the biased network $G_\text{bias}$ (even so
  for~\autoref{thm:typical-coherence-gnn},
  see~\autoref{rmk:typical-coherence-random-weights-with-bias}).
\end{remark}

With these ingredients, we provide a suggestive ``cartoon'' of the main
theoretical contribution of this work, which itself can be found
in~\autoref{thm:subIso-GCS}.

\begin{sketch}[Cartoon]
  \label{sketch:cartoon}
  Let $G$ be a $(k, d, n)$-generative network, and $A$ a subsampled
  isometry. Suppose $\range(G)-\range(G)$ is ``incoherent'' with respect to the
  rows of $A$, quantified by a parameter $\alpha > 0$. If the number of
  measurements $m$ satisfies $m \gtrsim kdn \alpha^2$ (up to log
  factors), then, with high probability on $A$, it is possible to
  approximately recover an unknown signal $x_0 \in \range(G)$ from noisy
  underdetermined linear measurements $b = Ax_0 + \eta$ with nearly order-optimal
  error.
\end{sketch}

The coherence parameter $\alpha$ is defined below in~\autoref{def:coherence}
using the measurement norm introduced in~\autoref{def:measurement-norm}. The
quantification of $\alpha$ is discussed thereafter, and fully elaborated
in~\autoref{sec:typical-coherence}. The notion of ``incoherence'' in the cartoon
above is specified in~\autoref{coro:G-G-RIP}. Coherence is related to the
concept of incoherent bases~\cite[p.~373]{foucart2017mathematical}, while the
measurement norm is closely related to the so-called $X$-norm
in~\cite{rudelson2008sparse}. Effectively, coherence characterizes the alignment
between the components comprising $\range(G)$ and the row vectors $A_{i}$ of the
subsampled isometry $A$.

\begin{definition}[Measurement norm]
  \label{def:measurement-norm}
  Let $U \in \complex^{n \times n}$ be a unitary matrix. Define the norm
  $\|\cdot\|_{U} : \complex^{n} \to [0, \infty)$ by
  \begin{align*}
    \|x\|_{U} := \|Ux\|_{\infty} = \max_{i \in [n]} \left| \ip{U_{i}, x} \right|. 
  \end{align*}
\end{definition}

\begin{definition}[Coherence]
  \label{def:coherence}
  Let $T \subseteq \reals^{n}$ be a set and
  $U \in \complex^{n\times n}$ a unitary matrix. For $\alpha > 0$, say
  that $T$ is $\alpha$-coherent with respect to $\|\cdot\|_U$ if
  \begin{align*}
    \sup_{x \in T\cap \sphn} \|x\|_U 
    \leq \alpha. 
  \end{align*}
  We refer to the quantity on the left-hand side as the coherence.
\end{definition}

The idea is that the structural proxy/prior under consideration should be
\emph{incoherent} with respect to the measurement process. Thus, we desire that
chords in $\range(G)$\footnote{We will need to expand this set slightly
  via~\autoref{def:minimal-subspace-covering}.} be not too closely aligned (in
the sense controlled by $\alpha$) with the rows of $U$, with which the
subsampled isometry $A$ is associated. This follows a paradigm from classic
compressed sensing: \textit{democracy of the measurement process}, \ie no single
measurement should be essential for signal recovery, rather the measurements
should be used jointly. This is natural, since we are randomly sampling, and
any potential measurement may not be sampled. The standard definition of
coherence in compressed sensing, and its nonrandom origins, involves a bound on
inner products between columns of the sensing matrix~\cite[Chapter
5]{foucart2017mathematical}; it applies to deterministic measurement matrices,
and is somewhat different than our definition. There is also a definition of
\textit{incoherence} in~\cite[Chapter~12]{foucart2017mathematical} (see
also~\cite{donoho2003optimally}) or incoherence
property~\cite{candes2011probabilistic}. The latter two are defined in the
setup of random sampling and are somewhat analogous to the parameter we defined.
 
Though it is likely difficult to measure coherence precisely in practice, we
propose a computationally efficient heuristic that upper bounds the coherence.
For perspective, we show that if $G$ has Gaussian weights, one may take
$\alpha^2 \sim kd/n$ in~\nameref{sketch:cartoon} (see~\autoref{thm:subIso-GCS}
and~\autoref{thm:typical-coherence-gnn}), thereby giving a sample complexity,
$m$, proportional to $(kd)^2$ up to log factors. We leave improving the
quadratic dependence as an open question, discussed further
in~\autoref{sec:conclusion}.

We briefly itemize the main contributions of this paper:
\begin{itemize}
\item we introduce the \emph{coherence} for characterizing recovery efficacy via
  the alignment of the network's range with the measurement matrix
  (see~\autoref{def:coherence});
\item we establish a restricted isometry property for $(k,d,n)$-generative
  networks with subsampled isometries (see~\autoref{thm:G-RIP}
  and~\autoref{coro:G-G-RIP});
\item we prove sample complexity and recovery bounds in this setting
  (see~\autoref{thm:subIso-GCS});
\item we propose a regularization strategy for training GNNs with low coherence
  (see~\autoref{sec:numerics-setup}) and demonstrate improved sample complexity
  for recovery (see~\autoref{sec:numerics-results});
\item together with our theory, we provide compelling numerical simulations that
  support coherence as a natural quantity of interest linked to favourable deep
  generative recovery (see~\autoref{sec:numerics-results}).
\end{itemize}

\subsection{Related work}
\label{sec:related-work}

Theoretically, \cite{bora2017compressed} have analyzed compressed sensing
problems in the so-called generative prior framework, focusing on Gaussian or
subgaussian measurement matrices. This led to much follow-up work in the
generative prior framework, albeit none in the subsampled Fourier setting to our
knowledge. For example, \cite{berk2021deep} extends the analysis to the setting
of demixing with subgaussian matrices, while~\cite{liu2022non} analyzes the
semi-parametric single-index model with generative prior under Gaussian
measurements. Finally, exact recovery of the underlying latent code for GNNs
(\ie seeking $z \in \reals^{k}$ such that $x = G(u)$) has been analyzed;
however, these analyses rely on the GNN having a suitable structure with weight
matrices that possess a suitable
randomness~\cite{hand2019global,joshi2021plugin,joshi2019global,
  hand2018phase}. For a review of these and related problems,
see~\cite{scarlett2022theoretical}.

Promising empirical results of~\cite{jalal2021robust} suggest remarkable
efficacy of generative compressed sensing (GCS) in realistic measurement
paradigms. Furthermore, the authors provide a framework with theoretical
guarantees for using Langevin dynamics to sample from a generative
prior. Several recent works have developed sophisticated generative adversarial
networks (GANs) (which are effectively a type of GNN) for compressed sensing in
medical imaging~\cite{deora2020structure, mardani2018deep}. Other work has
empirically explored multi-scale (non-Gaussian) sampling strategies for image
compressed sensing using GANs~\cite{li2022fast}. Separately,
see~\cite{wentz2022genmod} for the use of GCS in uncertainty quantification of
high-dimensional partial differential equations with random inputs. Recently
popular is the use of untrained GNNs for signal recovery~\cite{ulyanov2018deep,
  heckel2019deep}. For instance, \cite{darestani2021accelerated} executed a
promising empirical investigation of medical image compressed sensing using
untrained GNNs.

Compressed sensing with subsampled isometries is well studied for sparse signal
recovery. The original works developing such recovery guarantees
are~\cite{candes2006robust, donoho2006compressed}, with improvements appearing
in~\cite{rudelson2008sparse,
  rauhut2010compressive}. See~\cite{foucart2017mathematical} for a thorough
presentation of this material including relevant
background. See~\cite[Sec.~4]{dirksen2015tail} for a clear presentation of this
material via an extension of generic chaining. In this setting, the best-known
number of log factors in the sample complexity bound sufficient to achieve the
restricted isometry property is due to~\cite{bourgain2014improved} with
subsequent extensions and improvements in~\cite{haviv2017restricted,
  chkifa2018polynomial, brugiapaglia2021sparse}. \cite{naderiPlanSparsityFree}
address compressed sensing with subsampled isometries when the structural proxy
is a neural network with random weights.

Using a notion of coherence to analyze the solution of convex linear inverse
problems was proposed in~\cite{candes2006robust,
  candes2007sparsity}. \cite{cape2019two} relate this notion to the matrix norm
$\|\cdot\|_{2\to\infty}$ (defined in~\autoref{sec:notation}) in order to analyze
covariance estimation and singular subspace recovery. Additionally,
see~\cite{donoho2003optimally} or~\cite[p.~373]{foucart2017mathematical} for a
discussion of \emph{incoherent bases}, and~\cite[p.~1034]{rudelson2008sparse}
for the analogue of our measurement norm in the sparsity case.

The present work relies on important ideas from high-dimensional probability,
such as controlling the expected supremum of a random process on a geometric
set. These ideas are well treated in~\cite{liaw2017simple, jeong2020sub};
see~\cite{vershynin2018high} for a thorough treatment of high-dimensional
probability. This work also relies on counting linear regions comprising the
range of a $\relu$-activated {GNN}. In this respect, we rely on a result that
appears in~\cite{naderi2021beyond}. Tighter but less analytically tractable
bounds appear in~\cite{serra2018bounding}, while a computational exploration of
region counting has been performed in~\cite{novak2018sensitivity}.

\subsection{Notation}
\label{sec:notation}

For an integer $n \geq 1$ denote $[n] := \{1, \ldots, n\}$. For
$x \in \complex^{n}$, denote the $\ell_{p}$ norm for $1 \leq p < \infty$ by
$\|x\|_{p} := \left( \sum_{i=1}^{n} |x_{i}|^{p} \right)^{1/p}$ and for
$p = \infty$ by $\|x\|_{\infty} := \max_{i \in [n]} |x_{i}|$. Here, if
$x \in \complex$ then $|x| = \sqrt{\Re(x)^{2} + \Im(x)^{2}}$ and the conjugate
is given by $\bar x := \Re(x) - \i\Im(x)$. If $X \in \complex^{m \times n}$ is a
matrix then the conjugate transpose is denoted
$X^{*} = (\bar X_{ji})_{j \in [n], i \in [m]}$. The $\ell_{p}$ norm for real
numbers, $1 \leq p \leq \infty$ is defined in the standard, analogous
way. Denote the real and complex sphere each by
$\sphn := \{ x : \|x\|_{2} = 1\}$, disambiguating only where unclear from
context. The operator norm of a matrix $X \in \complex^{n \times n}$, induced by
the Euclidean norm, is denoted $\|X \| := \sup_{\|z\|_{2} = 1}
\|Xz\|_{2}$. Unless otherwise noted, $X_{i}$ denotes the $i$th row of the matrix
$X$, viewed as a column vector. The Frobenius norm of $X$ is denoted $\|X\|_F$
and satisfies $\|X\|_F^2 = \sum_{i=1}^m \|X_i\|_2^2$. The matrix norm
$\|\cdot\|_{p \to q}$ for $1 \leq p, q \leq \infty$ is
$\|X\|_{p \to q} := \sup_{z \neq 0} \frac{\|Xz\|_q}{\|z\|_p}$. We use
$\Pi_\mathcal{L}$ to denote the standard $\ell_2$ projection operator onto the
set $\mathcal{L}$, which selects a single point lexicographically, if necessary,
to ensure uniqueness. $\DistBer(p)$ denotes the Bernoulli distribution with
parameter $p$; $\DistBinom(n, p)$ the binomial distribution for $n$ items with
rate $p$.

Throughout this work, $C > 0$ represents an absolute constant having no
dependence on any parameters, whose value may change from one appearance to the
next. Constants with dependence on a parameter will be denoted with an
appropriate subscript --- \eg $C_{\delta}$ is an absolute constant depending
only on a parameter $\delta$. Likewise, for two quantities $a, b$, if
$a \lesssim b$ then $a \leq Cb$; analogously for $a \gtrsim b$. Finally, given
two sets $A, B \subseteq \reals^{n}$, $A\pm B$ denotes the Minkowski
sum/difference: $A \pm B := \{ a \pm b : a \in A, b \in B\}$. Similarly, for
$a \in \reals^{n}$, $a - B := \{ a - b : b \in B\}$ and
$a B := \{ ab : b \in B\}$. The range of a function $f : \reals^n \to \reals^m$
is denoted $\range(f) := \{f(x) : x \in \reals^n\}$ (\eg if $X$ is a matrix then
$\range(X)$ denotes the column space of $X$). As above,
$\sigma(x) := \max(x,0)$, which may act element-wise on a vector.

\section{Main results}
\label{sec:main-results}

Proofs of results in this section are deferred to~\autoref{sec:proofs-results}.

Observe, if $G$ is a $(k,d,n)$-generative network, then $\range(G)$ and
$\mathcal{G} := \range(G) - \range(G)$ are unions of polyhedral cones
(see~\autoref{lem:upper-bound-number-of-cones}
and~\autoref{rmk:extend-to-range-difference}). Note that polyhedral cones
(\autoref{def:polyhedral-cone}) are convex. We introduce the following
definition to expand each cone into a full subspace.
\begin{definition}
  \label{def:minimal-subspace-covering}
  Let $\mathcal{C} \subseteq \reals^{n}$ be the union of $N$ convex cones:
  $\mathcal{C} = \bigcup\limits_{i=1}^N \mathcal{C}_i$. Define the piecewise
  linear expansion
  \begin{align*}
    \Delta(\mathcal{C}) := \bigcup\limits_{i=1}^N \Span(\mathcal{C}_i) = \bigcup_{i=1}^N
    (\mathcal{C}_i - \mathcal{C}_i),
  \end{align*}
\end{definition}
The second equality follows
from~\autoref{prop:cone-difference-is-subspace}. See~\autoref{rmk:Delta-properties}
for a list of properties of $\Delta$, including uniqueness. Note each cone
comprising $\range(G)$ has dimension at most $k$, hence $\Delta(\range(G))$ is a
union of linear subspaces each having dimension at most $k$.

We now present the main result of the paper, which establishes sample complexity
and recovery bounds for generative compressed sensing with subsampled
isometries. Below, $x^\perp := x_{0} - \proj_{\range(G)} x_{0}$.

\begin{theorem}[subsampled isometry GCS]
  \label{thm:subIso-GCS}
  Let $G : \reals^{k} \to \reals^{n}$ be a $(k, d, n)$-generative network with
  layer widths $k = k_0 \leq k_1, \ldots, k_d$ where $k_d := n$,
  $\varepsilon, \hat\varepsilon > 0$, $\mathcal{G} := \range(G) - \range(G)$ and
  $A \in \complex^{\tilde m \times n}$ a subsampled isometry associated with a
  unitary matrix $U \in \complex^{n \times n}$. If $\Delta(\mathcal{G})$ is
  $\alpha$-coherent with respect to $\|\cdot\|_{U}$, and
  \begin{align*}
    m \gtrsim \alpha^{2} n %
    \left( 2k \sum_{i = 1}^{d-1} \log \left( \frac{2ek_{i}}{k} \right) %
    + \log \frac{4k}{\varepsilon} \right),
  \end{align*}
  then, the following holds with probability at least $1-\varepsilon$ on the
  realization of $A$.

  For any $x_{0} \in \reals^{n}$, let $b := Ax_{0} + \eta$ where
  $\eta \in \complex^{\tilde m}$. Let $\hat x \in \reals^n$ satisfy
  $\|A\hat x - b\|_2 \leq \min_{x \in \range(G)} \|Ax - b\|_{2} +
  \hat{\varepsilon}$. Then,
  \begin{align*}
    \|\hat x - x_0\|_2 %
    \leq \|x^\perp\|_2 + 3\|Ax^\perp\|_2 + 3 \|\eta\|_2 + \frac{3}{2}\hat\varepsilon.
  \end{align*}
\end{theorem}

\begin{remark}
  \label{rmk:theorem-ii-1-weakening}
  Since $\range(G)$ is a union of polyhedral cones,
  $\mathcal{G} \subseteq \Delta(\mathcal{G}) \subseteq \mathcal{G} -
  \mathcal{G}$. Hence, it is sufficient to assume that
  $\mathcal{G} - \mathcal{G}$ be $\alpha$-coherent with respect to
  $\|\cdot\|_{U}$. This containment may aid practitioners to control $\alpha$
  since one may sample from $\mathcal{G} - \mathcal{G}$.
\end{remark}

\begin{remark}
  \label{rmk:modelling-error}
  The approximation error $\|x^\perp\|_2$ is controlled by the
  \emph{expressivity} of $G$, satisfying (by definition)
  \begin{align*}
    \|x^\perp\|_2 = \min_{u \in \reals^k} \|G(u) - x_0\|_2.
  \end{align*}
  The modelling error incurred via $\|Ax^\perp\|_{2}$ could be large compared to
  $\|x^\perp\|_{2}$:
  $\|Ax^\perp\|_2 \leq \frac{\sqrt{n}}{\sqrt{m}}\|x^\perp\|_2$ in
  general. However, if $G$ admits a good representation of the modelled data
  distribution, then one might expect this term still to be small. Certainly, if
  $x_{0} \in \range(G)$, the final expression in~\autoref{thm:subIso-GCS}
  reduces to
  \begin{align*}
    \|\hat x - x_0\|_2 %
    \leq 3 \|\eta\|_2 + \frac{3}{2}\hat\varepsilon.
  \end{align*}
  Otherwise, if $x^\perp$ is independent of $A$,
  $\E \|Ax^\perp\|_2 \leq \|x^\perp\|_2$ by Jensen's inequality. Thus, by
  Markov's inequality one has
  $\pr \left\{ \|Ax^\perp\|_2 \geq \kappa \|x^\perp\|_2\right\} \leq
  \kappa^{-1}$. Finally, a strategy for more precisely controlling
  $\|Ax^\perp\|_2$ is given in~\autoref{sec:control-approx-error} (see
  especially~\autoref{prop:cramer-chernoff-concentration}).
\end{remark}

Analogous to the restricted isometry property of compressed sensing or the
set-restricted eigenvalue condition of~\cite{bora2017compressed}, the proof
of~\autoref{thm:subIso-GCS} relies on a restricted isometry condition. This
condition guarantees that pairwise distances of points in $\range(G)$ are
approximately preserved under the action of $A$. We first state a result
controlling norms of points in $\range(G)$ under the action of $A$; control over
pairwise distances then follows easily.

\begin{theorem}[Gen-RIP]
  \label{thm:G-RIP}
  Let $A \in \complex^{\tilde m \times n}$ be a subsampled isometry associated
  with a unitary matrix $U \in \complex^{n\times n}$ and
  $\delta,\varepsilon > 0$. Suppose that $G : \reals^{k} \to \reals^{n}$ is a
  $(k, d, n)$-generative network with layer widths
  $k = k_0 \leq k_1, \ldots, k_d$ where $k_d := n$ and that $\Delta(\range(G))$
  is $\alpha$-coherent with respect to $\|\cdot\|_{U}$. If
  \begin{align*}
    m %
    \gtrsim \frac{\alpha^{2} n}{\delta^{2}} %
    \left( k \sum_{i = 1}^{d-1} \log \left( \frac{2ek_{i}}{k} \right) %
    + \log \frac{2k}{\varepsilon} \right),
  \end{align*}
  then with probability at least $1 - \varepsilon$ on the realization of $A$, it
  holds that
  \begin{align*}
    \sup_{x \in \range(G) \cap \sphn} \left| \|Ax\|_{2} - 1 \right| \leq \delta. 
  \end{align*}
\end{theorem}

\begin{remark}
  \label{rmk:sample-complexity-ambient-dim}
  In~\autoref{sec:typical-coherence} we show that $\alpha$ can have dependence
  on $n$ proportional to $n^{-1/2}$, ignoring log factors
  (see~\autoref{prop:dist} and~\autoref{thm:typical-coherence-gnn}). Therefore,
  the sample complexity can be independent of the ambient dimension $n$ (again
  ignoring log factors).
\end{remark}

\begin{remark}
  Analogous to~\autoref{rmk:theorem-ii-1-weakening}, in~\autoref{thm:G-RIP} it
  is sufficient to assume $\alpha$-coherence of $\range(G) - \range(G)$, since
  $\range(G) \subseteq \Delta(\range(G)) \subseteq \range(G) - \range(G)$.
\end{remark}

We now state the result that provides the notion of restricted isometry needed
for~\autoref{thm:subIso-GCS}. This result, which controls pairwise differences
of elements in $\range(G)$, is an immediate consequence of~\autoref{thm:G-RIP}
using the observation in~\autoref{rmk:extend-to-range-difference}.

\begin{corollary}[Restricted isometry on the difference set]
  \label{coro:G-G-RIP}
  Let $G : \reals^{k}\to \reals^{n}$ be a $(k, d, n)$-generative network with
  layer widths $k = k_0 \leq k_1, \ldots, k_d$ where $k_d := n$,
  $\mathcal{G} := \range(G) - \range(G)$, $\delta,\varepsilon > 0$, and suppose
  $A \in \complex^{\tilde m \times n}$ is a subsampled isometry associated with
  a unitary matrix $U \in \complex^{n \times n}$. Assume that
  $\Delta(\mathcal{G})$ is $\alpha$-coherent with respect to $\|\cdot\|_{U}$. If
  \begin{align*}
    m %
    \gtrsim \frac{\alpha^{2} n}{\delta^{2}} %
    \left( 2k \sum_{i = 1}^{d-1} \log \left( \frac{2ek_{i}}{k} \right) %
    + \log \frac{4k}{\varepsilon} \right),
  \end{align*}
  then with probability at least $1 - \varepsilon$ on the realization of $A$, it holds that
  \begin{align*}
    \sup_{x \in \mathcal{G}\cap \sphn} \left| \|Ax\|_{2} - 1 \right| \leq \delta. 
  \end{align*}
\end{corollary}

\begin{remark}
  In fact, the proof of~\autoref{thm:G-RIP} yields the stronger restricted
  isometry bound
  \begin{align*}
    \sup_{x \in \Delta(\range(G)) \cap\sphn} \left| \|Ax\|_2 - 1\right| \leq \delta.
  \end{align*}
  Consequently, the restricted isometry bound in~\autoref{coro:G-G-RIP} can be strengthened to 
  \begin{align*}
    \sup_{x \in \Delta(\mathcal{G})\cap\sphn} \left|\|Ax\|_2 - 1\right| \leq \delta.
  \end{align*}
  \vskip-10pt
\end{remark}

The proofs of~\autoref{thm:G-RIP} and~\autoref{thm:subIso-GCS} are deferred
to~\autoref{sec:proofs-results}. The result~\autoref{thm:G-RIP} follows directly
from~\autoref{lem:subgaussian-control-on-subspace}
and~\autoref{lem:upper-bound-number-of-cones}, the former of which is presented
next. It characterizes restricted isometry of a subspace incoherent with
$\|\cdot\|_{U}$. Its proof is deferred to~\autoref{sec:proofs-results}.

\begin{lemma}[RIP for incoherent subspace]
  \label{lem:subgaussian-control-on-subspace}
  Let $A \in \complex^{\tilde m \times n}$ be a subsampled isometry associated
  with a unitary matrix $U \in \complex^{n\times n}$. Suppose that
  $\mathcal{L} \subseteq \reals^{n}$ is a $k$-dimensional subspace that is
  $\alpha$-coherent with respect to $\|\cdot\|_{U}$. Then, for any
  $0 \leq \delta \leq 1$,
  \begin{align*}
    \pr \left\{ \sup_{x \in \mathcal{L} \cap \sphn} \left| \|Ax\|_{2} - 1 \right| \geq \delta \right\} %
    \leq 2k \exp \left( - \frac{C\delta^{2}m}{\alpha^{2}n} \right). %
  \end{align*}
\end{lemma}

\begin{remark}
  Convincing empirical results of~\cite{novak2018sensitivity} suggest the number
  of linear regions for empirically observed neural networks may typically be
  linear in the number of nodes, rather than exponential in the width. Such a
  reduction would be a boon for the sample complexity obtained
  in~\autoref{thm:G-RIP}, which depends on the number of linear regions
  comprising $\range(G)$ (using~\autoref{lem:upper-bound-number-of-cones};
  see~\autoref{sec:proofs-results}).
\end{remark}

\section{Typical Coherence}
\label{sec:typical-coherence}

Proofs for results in this section are deferred to
\autoref{sec:proofs-coher}. The first result of this section establishes a lower
bound on the coherence parameter. Together with~\autoref{coro:G-G-RIP} this
yields a quadratic ``bottleneck'' on the sample complexity in terms of the
parameter $k$.

\begin{proposition}\label{prop:dist}
  For a unitary matrix $U \in \complex^{n \times n}$, any $k$-dimensional
  subspace $T \subseteq \reals^n$ has coherence with respect to $\|\cdot \|_{U}$
  of at least $\sqrt{\frac{k}{n}}$. Furthermore, this lower bound is tight.
\end{proposition}

Under mild assumptions, when the generative network has random weights one may
show that this is a typical coherence level between the network and the
measurement operator.

\begin{theorem}
  \label{thm:typical-coherence-gnn}
  Let $U \in \complex^{n\times n}$ be a unitary matrix and $G$ be a
  $(k, d, n)$-generative network with layer widths
  $k = k_0 \leq k_1, \ldots, k_d$ where $k_d := n$. Let the last weight matrix
  of $G$, $W^{(d)}$, be iid Gaussian: $W^{(d)}_{ij} \iid \mathcal{N}(0,1)$,
  $i \in [k_d], j \in [k_{d - 1}]$. Let all other weights be arbitrary and
  fixed. Then, for any $\gamma \geq 0$, it holds with probability at least
  $1-2\exp(- \gamma^2)$ that $\Delta(\range(G)-\range(G))$ is $\alpha$-coherent
  with respect to $\|\cdot\|_U$, where
  \begin{align*}
    \alpha %
    \lesssim \sqrt{\frac{k}{n}} + \sqrt{\frac{\log{n}}{n}} + \sqrt{\frac{k}{n}\sum\limits_{i=1}^{d-1} \log{\frac{2e k_i}{k}}} + \frac{\gamma}{\sqrt{n}}.
  \end{align*}
\end{theorem}

\begin{remark}
  \label{rmk:typical-coherence-random-weights}
  We briefly comment on the behaviour of the third term, which, we argue,
  dominates for the principal case of interest. Assume the layers have
  approximately constant size: \ie for two absolute constants $C_1, C_2 > 0$,
  \begin{align*}
    \forall \ell \in [d], \quad C_1 \leq \log \frac{e k_\ell}{k} \leq C_2.
  \end{align*}
  In this case, all terms in the sum in the third term will be of the same
  order, making this term have order $\mathcal{O}(\sqrt\frac{kd}{n})$. If we
  further make the reasonable assumption that $dk > \log(n)$, then the third
  term dominates all others, hence
  \begin{align*}
    \alpha = \mathcal{O} \left( \sqrt{\frac{kd}{n}} \right).
  \end{align*}
\end{remark}

\begin{remark}
  Using~\autoref{coro:G-G-RIP}
  and~\autoref{rmk:typical-coherence-random-weights}, one may take as the sample
  complexity for~\autoref{thm:subIso-GCS}, in the case of a $(k,d,n)$-generative
  network with Gaussian weights,
  \begin{align*}
    m \gtrsim  \frac{2k^2d}{\delta^2} \sum_{i=1}^{d-1} \log \left(\frac{2ek_i}{k}\right) + \frac{kd}{\delta^2}\log \frac{4k}{\varepsilon}.
  \end{align*}
  We note in passing that an argument specialized to random weights is given
  in~\cite{naderiPlanSparsityFree}, with an improved sample complexity. Our
  goal in this section is not to find the optimal sample complexity for random
  weights, but to show the average case behaviour of the parameter $\alpha$.
\end{remark}

\section{Numerics}
\label{sec:numerics}

In this section we explore the connection between coherence and recovery error
empirically, to suggest that coherence is indeed the salient quantity dictating
recovery error. In addition, we propose a regularization strategy to train low
coherence GNNs. This regularization strategy is new to our knowledge. The first
experiment illustrates a phase portrait that empirically shows dependence on a
coherence (proxy) and number of measurements for successful recovery. We also
show, for a fixed number of measurements, that the probability of recovery
failure increases with higher coherence (proxy). In the second experiment, we
use the novel regularization approach to show that fewer measurements are
required for signal recovery when a GNN is trained to have low coherence.

\subsection{Experimental methodology}
\label{sec:numerics-setup}

\subsubsection{Coherence heuristic and regularization}
\label{sec:coherence-heuristic}

Ideally, in these experiments, one would calculate the coherence of the network
exactly, via~\autoref{def:coherence}. However, computing coherence is likely
intractable in general. Instead, we use an upper bound on the coherence obtained
as follows. Let $G$ be a $(k, d, n)$-generative network and let $W = W^{(d)}$ be
its final weight matrix. Write the QR decomposition of $W$ as
\begin{align*}
    W &= QR, & Q &:= \bmat{cc}{Q_1 & Q_2}, & R &:= \bmat{c}{R_1\\ 0},
\end{align*}
where $Q \in \reals^{n\times n}$ is orthogonal,
$R \in \reals^{n \times \tilde k}$ has invertible submatrix
$R_1 \in \reals^{\tilde k \times \tilde k}$ and
$Q_1 \in \reals^{\tilde n \times \tilde k}$ is the submatrix multiplying with
$R_1$. Let $\mathcal{G} := \range(G) - \range(G)$,
$\mathcal{W} := \range(W) \cap \sphn$ and let $D \in \reals^{n\times n}$ be an
orthogonal matrix. Using that
$\Delta(\mathcal{G})\cap\sphn \subseteq\mathcal{W}$, we bound the coherence with
respect to $\|\cdot \|_D$ as
\begin{align}
    \sup_{x \in \Delta(\mathcal{G})\cap\sphn} \|Dx\|_\infty %
    & \leq \sup_{x \in \mathcal{W}} \|Dx\|_\infty %
    \nonumber
    \\
    & = \max_{i \in [n]} \sup_{z \in \reals^{\tilde k}} \left\{ \left|D_i^\top Q R z \right| : \|Rz\|_2 = 1 \right\} %
    \nonumber
    \\
    & = \max_{i \in [n]} \sup_{v \in \reals^{\tilde k}} \left\{ \left|D_i^\top Q_1 v \right| : \|v\|_2 = 1 \right\} %
    \nonumber
    \\
    \label{eq:coherence-bound}
    & = \max_{i \in [n]} \left\|Q_1^\top D_i \right\|_2 = \left\| D Q_1 \right\|_{2\to \infty}, 
\end{align}
where the penultimate line uses $z := R_1^{-1} v$. To re-phrase:
$\Delta(\mathcal{G})$ is always $\|DQ_1\|_{2\to\infty}$-coherent with respect to
$\|\cdot\|_D$. Our experiments and theory are consistent with the hypothesis
that this is an effective heuristic for coherence.

Motivated by~\eqref{eq:coherence-bound}, we propose a strategy --- novel, to our
knowledge --- to promote low coherence of the final layer $W$ with respect to a
fixed orthogonal matrix $D$. This is achieved by applying the following
regularization $\rho$ to the final weight matrix of the GNN during training:
\begin{align}
  \label{eq:regularizer}
  \rho(W) = \left\| D W \right\|_{2\to \infty} + \lambda\|W^{\top} W -I\|_F.
\end{align}  
Namely, the regularizer $\rho$, with a fixed regularization parameter
$\lambda \geq 0$, is added to the training loss function. Roughly, this
regularization promotes low coherence because $\|W^\top W - I\|_F$ is smallest
when $W$ is orthonormal, making $\left\| D W \right\|_{2\to \infty}$ the
coherence of $\range(W)$ with respect to $D$.

\subsubsection{Network architectures}
\label{sec:numerics-architectures}

In the experiments, we use three generative neural networks trained on the MNIST
dataset~\cite{deng2012mnist}, which consists of 60,000 $28\times 28$ images of
handwritten digits. The GNNs are fully connected networks with three layers and
parameters $k=20$, $k_1 = k_2 = 500$, $n=784$. Precisely, let the first one be
$G^{(1)} = \sigmoid(W^{(1,3)}\sigma(W^{(1,2)}\sigma(W^{(1,1)}z)))$, where
$\sigmoid(x) = (1 + \exp(-x))^{-1}$ is the sigmoid activation function. Let the
remaining two GNNs be
$G^{(i)}(z) = W^{(i,3)}\sigma(W^{(i,2)}\sigma(W^{(i,1)}z))$, $i = 2, 3$. We use
$G^{(1)}$, which has a more realistic architecture for real applications, as a
point of comparison with $G^{(i)}, i = 2,3$. Variational autoencoders
(VAEs)~\cite{kingma2013auto}, with the decoder network as $G^{(1)}$ and
$G^{(2)}$, were trained using the Adam optimizer~\cite{kingma2014adam} with a
learning rate of 0.001 and a mini-batch size of 64 using
Flux~\cite{innes2018flux}. We trained another VAE with decoder network
$G^{(3)}$, using the same hyperparameters but using the regularization strategy
described in~\autoref{sec:coherence-heuristic} to promote low coherence of the
final layer $W^{(3,3)}$ with respect to a fixed orthogonal matrix
$D$. Specifically, the expression $10^4 \rho(W^{(3,3)})$, with $\lambda$ set to
1, was added to the VAE loss function. In all cases the VAE loss function was
the usual one. See~\cite{code-repo} for specific implementation details
including the definition of the encoders, and refer to~\cite{kingma2013auto,
  kingma2019introduction} for further background on VAEs.

 \subsubsection{Measurement matrix}
 \label{sec:numerics-measurement-matrix}
 
 Throughout the experiments, the matrix $D$ was chosen to be the discrete cosine
 transform (DCT) matrix. For DCT implementation details, see for
 instance~\cite[\texttt{fftpack.dct}]{virtanen2020scipy}. The matrix $A$ is a
 slight variation of the subsampled isometry defined
 in~\autoref{def:subsampled-isometry}, modified to ensure that each realization
 of $A$ has $m$ rows. Namely, the random matrix $A$ is subsampled from $D$ by
 selecting the first $m$ elements of a uniform random permutation of $[n]$. Note
 $A$ is still re-normalized as in~\autoref{def:subsampled-isometry}.

\subsubsection{First experiment}
\label{sec:first-experiment-methodology}

For the first experiment, let $G_\beta$ be a $(k, 2, n)$-generative network with
inner layers $W^{(i)} = W^{(1,i)}$, $i = 1,2$ and last layer
$W^{(3)} = W_\beta \in \reals^{n \times \tilde k}$ defined by
\begin{align*}
  W_\beta :=  \beta  W^{(1,3)} + (1-\beta) W^{(3,3)}.
\end{align*}
Recall that $W^{(1,3)}$ and $W^{(3,3)}$ are the final layers of $G^{(1)}$ and
$G^{(3)}$, respectively. Here, $\beta \in [0, 1]$ is an interpolation
parameter. The coherence, which was computed via~\eqref{eq:coherence-bound}, of
$\range(W^{(1,3)})$ was $0.98$, while the coherence of $\range(W^{(3,3)})$ was
$0.82$. As a result, for large $\beta$, one should expect $W_\beta$ to have
large coherence with respect to $\|\cdot\|_D$. We randomly sample
$z_0 \in \mathbb{R}^k$, fix the number of measurements
$m \in \{40,60,\ldots, 440\}$, and set $b = AG_\beta(z^0)$. For each measurement
size $m$ and coherence upper bound, we perform $20$ independent trials. For each
trial, we approximately solve~\eqref{eq:gnn-opt} by running ADAM with a learning
rate of $0.1$ for $5000$ iterations, or until the norm of the gradient is less
than $10^{-7}$, and set $\hat{z}$ to be the output. See~\cite{code-repo} for
specific implementation details. We say the target signal $G_\beta(z_0)$ was
successfully recovered if the relative reconstruction error (rre) between
$G_\beta(z_0)$ and $G_\beta(\hat{z})$ is less than $10^{-5}$:
\begin{align*}
  \rre(x_0, \hat x) := \frac{\|x_0 - \hat x\|_2}{\|x_0\|_2}.
\end{align*}

\subsubsection{Second experiment}
\label{sec:second-experiment-methodology}

For the second experiment, we use each trained network $G^{(i)}$, $i =
1,2,3$. The coherence upper bounds of
$\Delta\left(\range(G^{(2)}) - \range(G^{(2)})\right)$ and
$\Delta\left(\range(G^{(3)}) - \range(G^{(3)})\right)$, computed
using~\eqref{eq:coherence-bound} are $0.96$ and $0.81$, respectively, which
empirically shows that the regularization~\eqref{eq:regularizer} promotes low
coherence during training. For the networks $G^{(i)}$, let
$E^{(i)}:\reals^n\rightarrow\reals^k$ be the corresponding encoder network from
their shared {VAE}. We randomly sample an image $x^\sharp$ from the test set of
the MNIST dataset and let $x_0^{(i)} = G^{(i)}(E^{(i)}(x^\sharp))$ --- \ie
$x_0^{(i)} \in \range(G^{(i)})$ most likely resembles the test set image
$x^\sharp$. Let $m \in \{10,15,20,25,50,100,200,250\}$ and set
$b^{(i)} = Ax_0^{(i)}$. For each measurement size $m$, we run $10$ independent
trials. On each trial, we generate a realization of $A$ and randomly sample a
test image $x^\sharp$ from the MNIST dataset. To estimate $x_0^{(i)}$ on each
trial, we approximately solve~\eqref{eq:gnn-opt} by running ADAM with a learning
rate of $0.1$ for $5000$ iterations, or until the Euclidean norm of the gradient
is less than $10^{-7}$. See~\cite{code-repo} for specific implementation
details.

\subsection{Numerical results}
\label{sec:numerics-results}

\subsubsection{Recovery phase transition}
\label{sec:first-experiment}

\begin{figure}[t]
  \centering
  \begin{subfigure}{\linewidth}
    \centering
    \includegraphics[width=.9\linewidth]{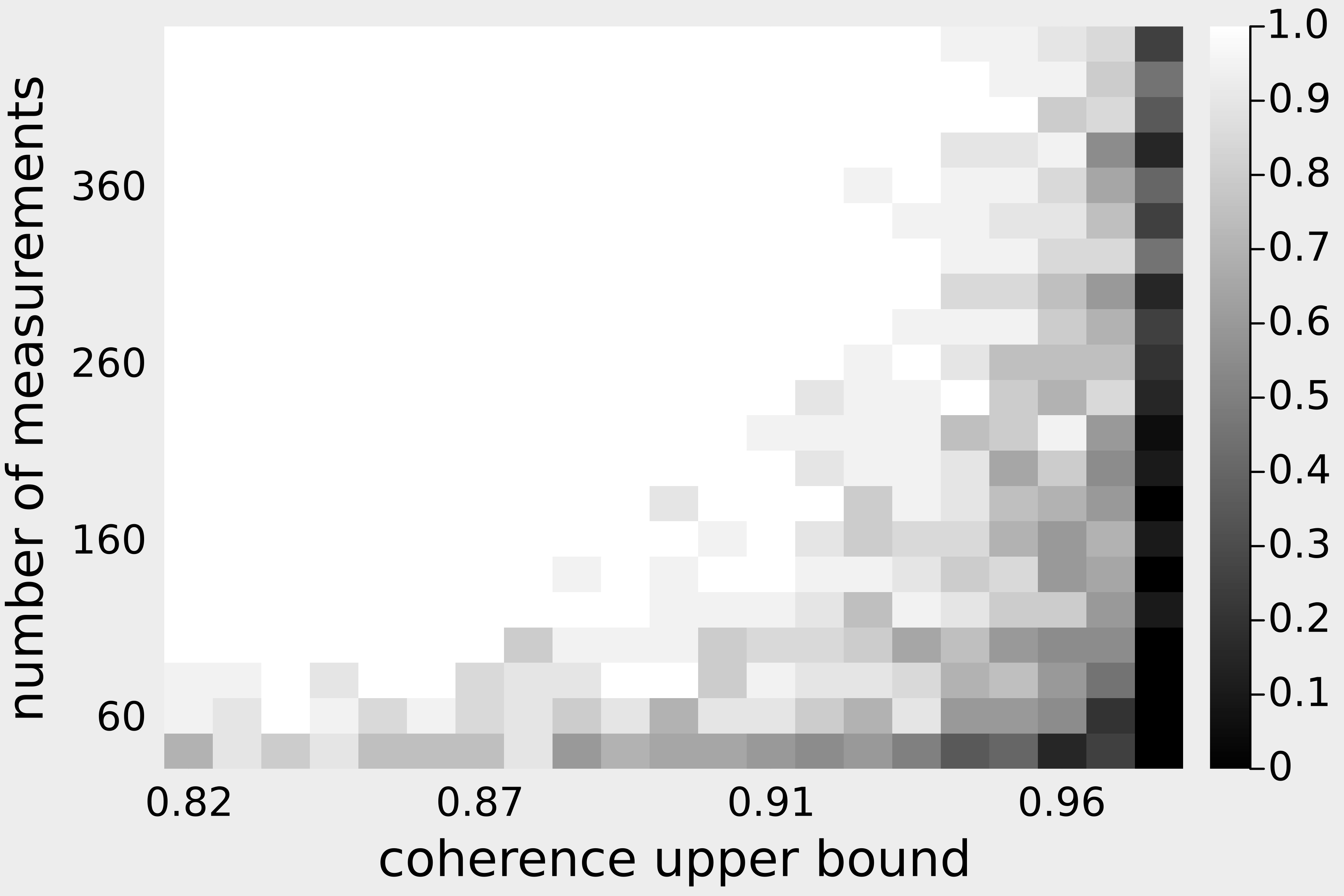} 
    \caption{Empirical recovery probability as a function of coherence and
      $m$. Each block corresponds to the average from 20 independent
      trials. White corresponds with 20 successful recoveries (rre
      $\leq 10^{-5}$); black with no successful
      recoveries.\label{fig:phaseplot}}
  \end{subfigure}
  
  \begin{subfigure}{\linewidth}
    \centering 
    \includegraphics[width=.9\linewidth]{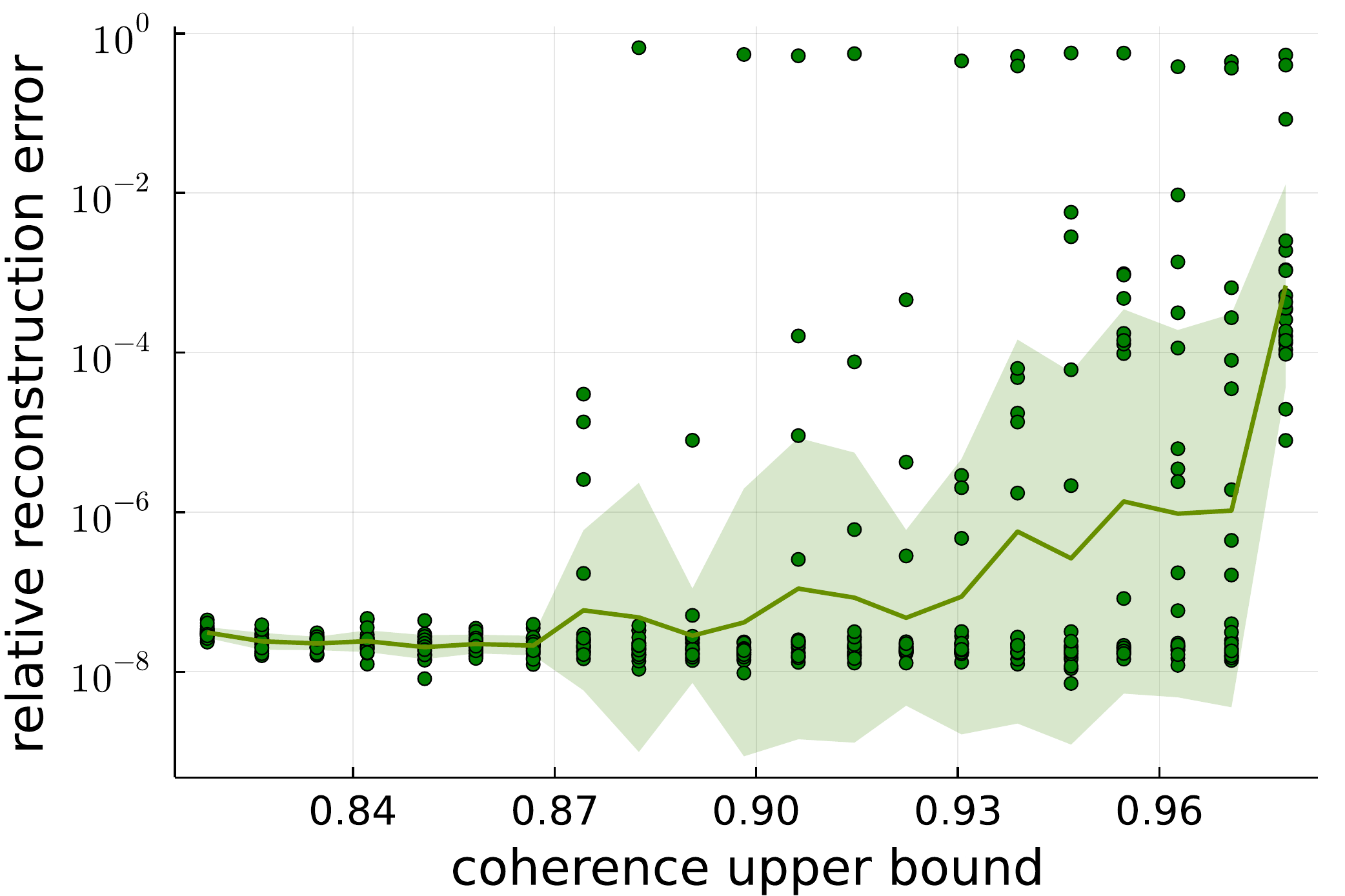}
    \caption{Empirical rre as a function of coherence for $m = 100$. Each dot
      corresponds to one of 20 trials at each coherence level. The solid line
      shows the empirical geometric mean rre \emph{vs.}\ coherence upper
      bound. The envelope shows $1$ geometric standard
      deviation.\label{fig:error-coherence}}
  \end{subfigure}
  \caption{Dependence of recovery on coherence and number of measurements $m$ for GNNs trained on MNIST.\label{fig:MNIST-recovery-dependence}}
\end{figure}

The results of the first experiment appear
in~\autoref{fig:MNIST-recovery-dependence}. Specifically,
\autoref{fig:phaseplot} plots the fraction of successful recoveries from 20
independent trials as a function of the coherence
heuristic~\eqref{eq:coherence-bound} and number of measurements. White squares
correspond to $100\%$ successful recovery (all errors were below $10^{-5}$),
while black squares correspond to no successful recoveries (all errors were
above $10^{-5}$). In~\autoref{fig:error-coherence}, we show a slice of the phase
plot for $m = 100$, plotting rre as a function of the coherence
heuristic~\eqref{eq:coherence-bound}. Each dot corresponds to one of 20 trials
at each coherence level. The plot is shown on a $\log$-$y$ scale. The solid line
plots the geometric mean of rre as a function of coherence, with an envelope
representing $1$ geometric standard deviation
(see~\cite[App.~A.1.3]{adcock2022sparse} for more information on this
visualization strategy). \autoref{fig:MNIST-recovery-dependence} indicates that
coherence may be effectively controlled via the
heuristic~\eqref{eq:coherence-bound}, and that coherence is a natural quantity
associated with recovery performance. These findings corroborate our theoretical
results.
 
\subsubsection{Incoherent networks require fewer measurements}
\label{sec:second-experiment}

In the second experiment, we provide compelling numerical simulations that
support our regularization strategy for lowering coherence of the trained
network, resulting in stable recovery of the target signal with much fewer
measurements. The results of the second experiment are shown
in~\autoref{fig:MNIST_recovery}
and~\autoref{fig:MNIST_random_comp}. In~\autoref{fig:MNIST_recovery}, we show
the recovered image for three images from the MNIST test set.
\begin{figure}[t]
  \centering
  \includegraphics[width=.9\linewidth]{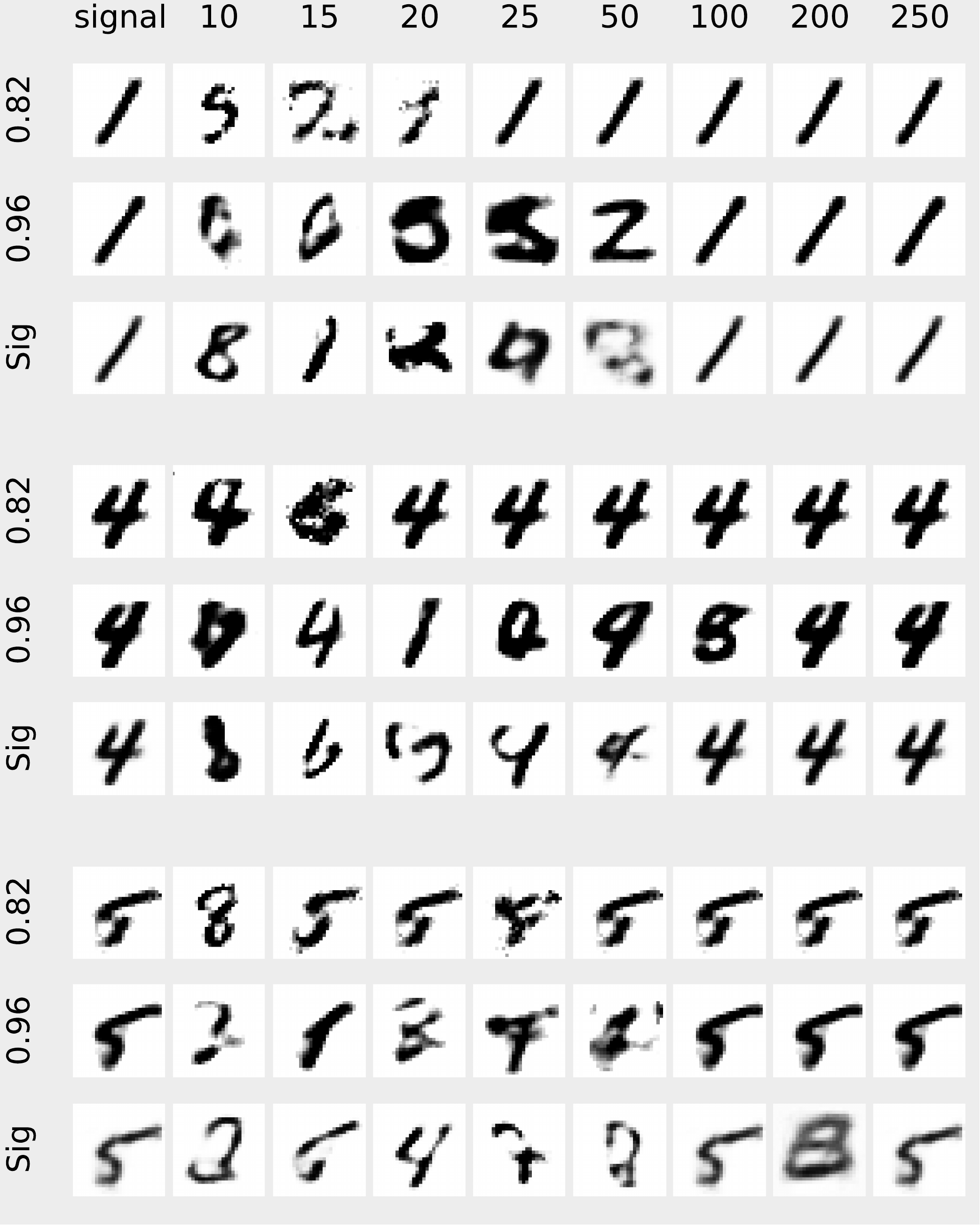}
  \caption{Recovery comparison of MNIST images for various measurement sizes $m$
    (denoted by column heading) for a low coherence network, high-coherence
    network and network with final sigmoid activation. In each block: the top
    row corresponds to $G^{(3)}$ ($\alpha=0.82$); middle row $G^{(2)}$
    ($\alpha = 0.96$); bottom row $G^{(1)}$ (labelled Sig). The leftmost column,
    signal, corresponds to the target image $x_0^{(i)} \in
    \range(G^{(i)})$. \label{fig:MNIST_recovery}}
\end{figure}
\begin{figure}[t]
  \centering
  \includegraphics[width=.9\linewidth]{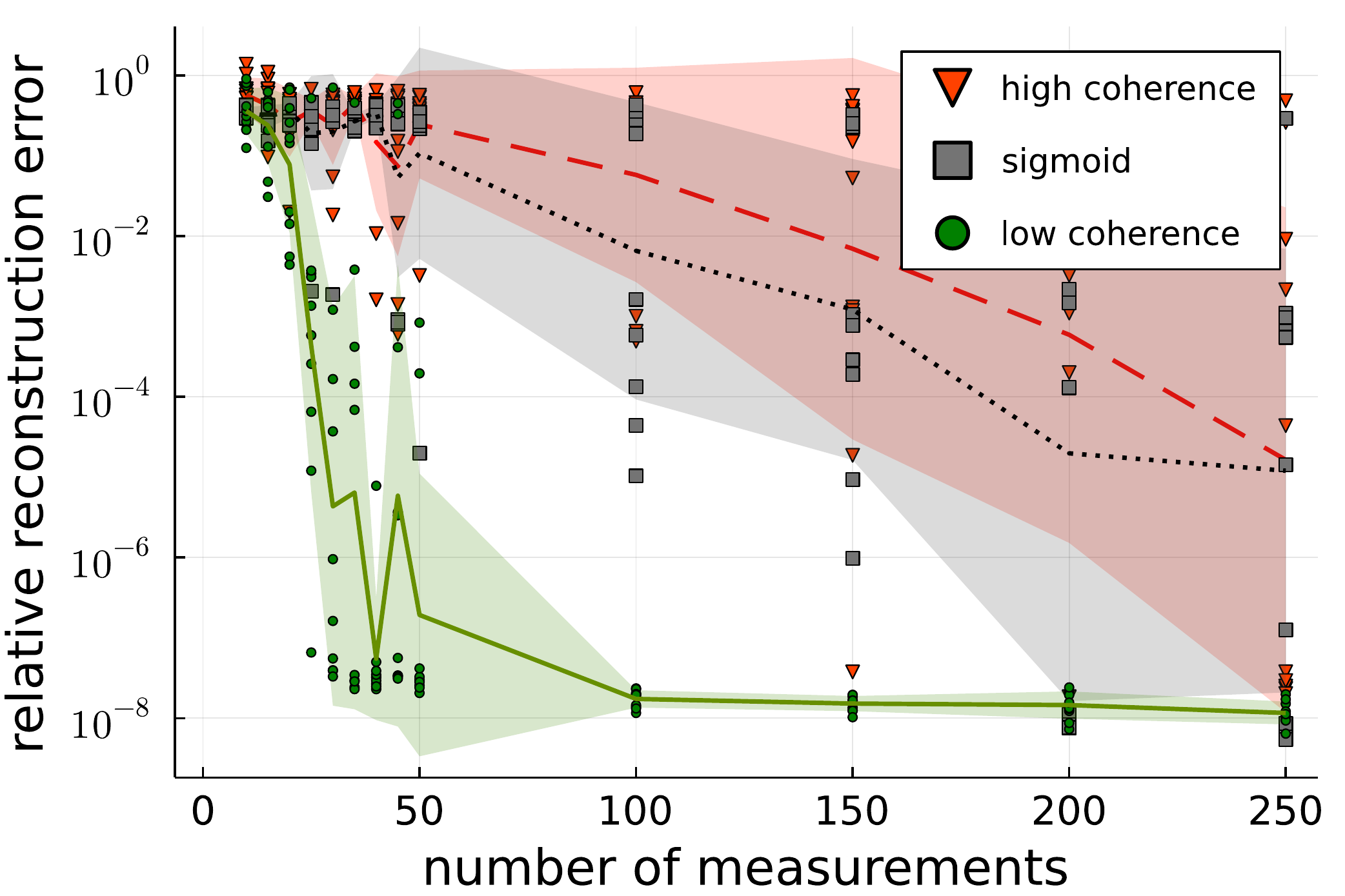}
  \caption{Performance comparison for three GNNs trained on the MNIST dataset
    --- one with low coherence, another with high coherence, and the last with
    sigmoid activation on the last layer. Plotted against number of measurements
    $m$ is $\rre$. For each value of $m$, each dot corresponds to one of 10
    trials. In each trial, a random image drawn from the MNIST test partition
    was used as the target signal. Lines depict the empirical geometric mean rre
    as a function of $m$; shaded regions correspond to $1$ geometric standard
    deviation. The solid line corresponds to $G^{(3)}$; the dashed line to
    $G^{(2)}$; the dotted line to $G^{(1)}$.\label{fig:MNIST_random_comp}}
\end{figure}
For each block of $3 \times 9$ images, the top row corresponds with the low
coherence $G^{(3)}$ ($\alpha = 0.82$); the middle row, the high coherence
$G^{(2)}$ ($\alpha = 0.96$); and the bottom row, $G^{(1)}$, which uses sigmoid
activation. The left-most column is the target image belonging to
$\range(G^{(i)})$, labelled \emph{signal}. All images were \emph{clamped} to the
range $[0,1]$. The figure shows that a GNN with low coherence can effectively
recover the target signal with much fewer measurements compared to a network
with high coherence, even when that network uses a final sigmoid activation
function (which is a realistic choice in practical settings). Remarkably, in
some cases we observed that images could be recovered with fewer than $k$
measurements. This highlights the importance of regularizing for networks with
low coherence during training. In~\autoref{fig:MNIST_random_comp}, we further
provide empirical evidence of the benefit of low coherence for recovery. For
each measurement, we show the results of 10 independent trials for $G^{(1)}$
(squares), $G^{(2)}$ (triangles) and $G^{(3)}$ (circles), respectively. The
lines correspond to the empirical geometric mean rre for each network: the
dotted line is associated with $G^{(1)}$; the dashed line, the high coherence
$G^{(2)}$ ($\alpha = 0.96$); and the solid line, the low coherence $G^{(3)}$
($\alpha=0.82$). The data are plotted on a $\log$-$y$ scale. Each shaded
envelope corresponds to $1$ geometric standard deviation about the respective
geometric means. This figure empirically supports that high probability
successful recovery is achieved with markedly lower sample complexity for the
lower coherence network $G^{(3)}$, as compared with either $G^{(2)}$ or
$G^{(1)}$. This finding corroborates our theoretical results.

\section{Proofs}
\label{sec:proofs}

\subsection[Proofs for main results]{{Proofs for~\nameref{sec:main-results}}}
\label{sec:proofs-results}

We proceed by proving~\autoref{thm:G-RIP}, then~\autoref{thm:subIso-GCS}. Note
that~\autoref{coro:G-G-RIP}, needed for~\autoref{thm:subIso-GCS}, follows
immediately from~\autoref{thm:G-RIP}
using~\autoref{rmk:extend-to-range-difference}.

\begin{proof}[Proof of~{\autoref{thm:G-RIP}}]
  By construction, $\range(G) \subseteq \Delta(\range(G))$; the latter set is a
  union of linear subspaces (see~\autoref{def:minimal-subspace-covering}
  and~\autoref{rmk:Delta-properties}). By~\autoref{lem:upper-bound-number-of-cones},
  $\range(G)$ is a union of no more than $N$ polyhedral cones of dimension at
  most $k$, with $N$ satisfying
  \begin{align*}
    \log N \leq k \sum_{i = 1}^{d-1} \log \left( \frac{2ek_{i}}{k} \right)
  \end{align*}
  via~\autoref{rmk:log-subspace-number-expression}. In particular,
  $\Delta(\range(G))$ is a collection of at most $N$ subspaces. For any linear
  subspace $\mathcal{L} \in \Delta(\range(G))$, observe that $\mathcal{L}$ is
  $\alpha$-coherent with respect to $\|\cdot\|_{U}$ by assumption. Consequently,
  by a union bound and application
  of~\autoref{lem:subgaussian-control-on-subspace},
  \begin{align*}
    &\pr \left\{ \sup_{x \in \range(G) \cap \sphn} \left|  \|Ax\|_{2} - 1 \right| \geq \delta \right\} %
    \\
    & \leq \sum_{\mathcal{L} \in \Delta(\range(G))} \pr \left\{ \sup_{x \in \mathcal{L} \cap \sphn} \left| \|Ax\|_{2} - 1 \right| \geq \delta \right\} %
    \\
    & \leq 2 N  k \exp \left( - \frac{C m \delta^{2}}{\alpha^{2}n} \right).
  \end{align*}
  The latter quantity is bounded above by $\varepsilon$ if
  \begin{align*}
    m \gtrsim \frac{\alpha^{2} n}{\delta^{2}} \left( \log N + \log \frac{2k}{\varepsilon} \right),
  \end{align*}
  whence, by substituting the bound for $\log N$, it suffices to take
  \begin{align*}
    m \gtrsim \frac{\alpha^{2} n}{\delta^{2}} \left( k \sum_{i = 1}^{d-1} \log \left( \frac{2ek_{i}}{k} \right) + \log \frac{2k}{\varepsilon} \right). 
  \end{align*}
\end{proof}

\begin{proof}[{Proof of~\autoref{thm:subIso-GCS}}]
  Recall $x^\perp := x_{0} - \proj_{\range(G)} x_{0}$. By triangle inequality
  and the observation that $\proj_{\range(G)}(x_0) \in \range(G)$,
  \begin{align*}
    \|A\hat x - b\|_2 %
    &\leq \min_{x\in\range(G)} \|Ax - b\|_2 + \hat\varepsilon %
    \\
    &\leq \|A\proj_{\range(G)}(x_0) - b\|_2 + \hat \varepsilon
    \\
    & = \|Ax^\perp + \eta\|_2 + \hat \varepsilon %
    \\
    & \leq \|Ax^\perp\|_2 + \|\eta\|_2 + \hat\varepsilon.
  \end{align*}
  Moreover, with probability at least $1 - \varepsilon$ on the realization of
  $A$, $A$ satisfies a restricted isometry condition on the difference set
  $\range(G) - \range(G)$ by~\autoref{coro:G-G-RIP}. Therefore, since
  $\hat x, \proj_{\range(G)}(x_0) \in \range(G)$,
  \begin{align*}
    & \|A \hat x - b\|_2 %
    \\
    & = \left\|A \left(\hat x - \proj_{\range(G)}(x_0)\right) - A\left(x_0 - \proj_{\range(G)}(x_0)\right) - \eta \right\|_2 %
    \\
    & \geq \left\|A\left(\hat x - \proj_{\range(G)}(x_0)\right)\right\|_2 - \|Ax^\perp\|_2 - \|\eta\|_2 %
    \\
    & \geq (1 - \delta) \left\| \hat x - \proj_{\range(G)}(x_0) \right\|_2 - \|Ax^\perp\|_2 - \|\eta \|_2.
  \end{align*}
  Assembling the two inequalities gives 
  \begin{align*}
      \left\| \hat x - \proj_{\range(G)}(x_0)\right\|_2 \leq \frac{1}{1 - \delta} \left[2 \|Ax^\perp\|_2 + 2\|\eta\|_2 + \hat\varepsilon\right]. 
  \end{align*}
  Finally, apply triangle inequality and choose $\delta = \frac{1}{3}$ to get
  \begin{align*}
    \|\hat x - x_0\|_2 %
    & \leq \|x_0 - \proj_{\range(G)}(x_0)\|_2 + \|\hat x - \proj_{\range(G)}(x_0)\|_2
    \\
    & \leq \|x^\perp\|_2 + 3 \|Ax^\perp\|_2 + 3\|\eta\|_2 + \frac{3}{2}\hat\varepsilon.
  \end{align*}
\end{proof}

\begin{proof}[Proof of~{\autoref{lem:subgaussian-control-on-subspace}}]
  Observe that $I = \sum_{i=1}^{n} U_{i}U_{i}^{*}$ since $U$ is a unitary
  matrix. Thus, since $\proj_{\mathcal{L}}^{2} = \proj_{\mathcal{L}}$,
  \begin{align*}
    & \sup_{x \in \mathcal{L} \cap \sphn} \left| \|Ax\|_{2}^{2} - 1 \right| %
    \\
    & = \sup_{x \in\mathcal{L}\cap\sphn} \left| x^{\top} (A^{*}A - I) x \right| %
    \\
    & = \sup_{x \in\mathcal{L}\cap\sphn} \left| x^{\top} \proj_{\mathcal{L}}^{\top} \left( \frac{n}{m} \sum_{i = 1}^{n} \theta_{i} U_{i}U_{i}^{*} - \sum_{i = 1}^{n} U_{i}U_{i}^{*} \right) \proj_{\mathcal{L}} x \right|.
  \end{align*}
  Define
  $\tilde U_{i} := \proj_{\mathcal{L}} U_{i} = \proj_{\mathcal{L}}^{\top}
  U_{i}$, using that $\proj_{\mathcal{L}}$ is an orthogonal projection, hence
  symmetric. Then,
  \begin{align*}
    & \sup_{x \in \mathcal{L} \cap \sphn} \left| \|Ax\|_{2}^{2} - 1 \right| %
    \\
    & = \sup_{x \in \mathcal{L} \cap \sphn} \left| x^{\top} \frac{n}{m} \left( \sum_{i = 1}^{n} \left(\theta_{i} - \frac{m}{n}\right) \tilde U_{i} \tilde U_{i}^{*} \right)x \right|. 
  \end{align*}
  Since $x$ and $\tilde{U}_{i}$ belong to a $k$-dimensional subspace, there
  exists a linear isometric embedding
  $E : \reals^k \to \mathcal{L} \subseteq{\reals^n}$ such that $x = E\check{x}$
  and $\tilde{U}_i = E\check{U}_{i}$ with
  $\check{x}, \check{U}_{i} \in \reals^k$. Thus,
  \begin{align*}
    & \sup_{x \in \mathcal{L} \cap \sphn} \left| x^{\top} \frac{n}{m} \left( \sum_{i = 1}^{n} \left(\theta_{i} - \frac{m}{n}\right) \tilde U_{i} \tilde U_{i}^{*} \right)x \right| 
    \\
    & = \sup_{x \in \sphn[k]} \left| x^{\top} E^{\top} \frac{n}{m} \left( \sum_{i = 1}^{n} \left(\theta_{i} - \frac{m}{n}\right) E \check{U}_{i} \check{U}_{i}^{*} E^{\top} \right)Ex \right| 
    \\
    & = \sup_{x \in \sphn[k]} \left| x^{\top} \frac{n}{m} \left( \sum_{i = 1}^{n} \left(\theta_{i} - \frac{m}{n}\right) \check{U}_{i} \check{U}_{i}^{*} \right)x \right|
    \\
    & = \left\| \frac{n}{m} \sum_{i = 1}^{n} \left(\theta_{i} - \frac{m}{n}\right) \check{U}_{i} \check{U}_{i}^{*} \right\|,
  \end{align*}
  where, in this case, $\|\cdot\|$ is the operator norm on Hermitian matrices
  over $\reals^k$ induced by the Euclidean norm. We will apply the matrix
  Bernstein inequality (\autoref{lem:matrix-bernstein}) to achieve concentration
  of the operator norm of the sum of mean-zero random matrices above. By the
  $\alpha$-coherence assumption,
  $\|\check U_{i}\|_2^2 = \|\tilde U_{i}\|_{2}^{2} %
  = \sup_{x \in \mathcal{L} \cap \sphn} \left| \ip{x, U_{i}} \right|^{2} %
  \leq \alpha^{2}$. Consequently, the operator norm of each constituent matrix
  is bounded almost surely: for each $i \in [n]$,
  \begin{align*}
    \left\| \frac{n}{m}\check U_{i} \check U_{i}^{*} \left(\theta_{i} - \frac{m}{n}\right) \right\| %
    \leq \frac{n}{m} \left\| \check U_{i} \check U_{i}^{*} \right\| %
    = \frac{n}{m}\|\tilde U_{i}\|_{2}^{2} %
    \leq \frac{n}{m} \alpha^{2}
  \end{align*}
  and the operator norm of the covariance matrix is bounded as
  \begin{align*}
    & \left\| \frac{n^2}{m^{2}} \sum_{i = 1}^{n} \|\check U_{i}\|_{2}^{2} \check U_{i} \check U_{i}^{*} \E \left(\theta_{i} - \frac{m}{n}\right)^{2} \right\| %
    \\
    & = \left\| \frac{n^2}{m^{2}}\sum_{i = 1}^{n} \|\check U_{i}\|_{2}^{2} \check{U}_{i} \check{U}_{i}^{*} \frac{m}{n}\left(1 - \frac{m}{n}\right) \right\| %
    \leq \alpha^{2} \left( \frac{n}{m} - 1 \right).
  \end{align*}
  Therefore, by~\autoref{lem:matrix-bernstein} it follows that
  \begin{align*}
    &\pr \left\{ \left\| \frac{n}{m}\sum_{i = 1}^{n}     \check{U}_{i} \check{U}_{i}^{*} \left(\theta_{i} - \frac{m}{n}\right) \right\| \geq \delta \right\} %
    \\
    & \leq 2k \exp \left( - \frac{m\delta^{2}/2}{n\alpha^{2}\left( 1 - \frac{m}{n} + \frac{\delta}{3} \right)} \right) %
    \\
    & \leq 2k \exp \left( - C\cdot \min \left\{ \frac{m\delta^{2}}{n\alpha^{2}(1 - m/n)}, \frac{m\delta}{n\alpha^{2}} \right\} \right). %
  \end{align*}

  To complete the proof, we adapt the argument from the proof
  of~\cite[Theorem~3.1.1]{vershynin2018high}. Indeed, for $\delta, z \geq 0$
  note that $|1 - z| > \delta \implies |z^{2} - 1| > \max(\delta, \delta^{2})$
  yields the implication
  $\max_{i} |1 - z_{i}| > \delta \implies \max_{i}|z_{i}^{2} - 1| > \max(\delta,
  \delta^{2})$. Consequently,
  \begin{align*}
    &\pr \left\{ \sup_{x \in \mathcal{L} \cap \sphn} \left| \|Ax\|_{2} - 1 \right| \geq \delta \right\} %
    \\
    & \leq \pr \left\{ \sup_{x \in \mathcal{L} \cap \sphn} \left| \|Ax\|_{2}^{2} - 1 \right| \geq \max(\delta, \delta^{2}) \right\} %
    \\
    & \leq 2k \exp \left( - \frac{C\delta^{2}m}{\alpha^{2}n } \right).
  \end{align*}
\end{proof}

\subsection[Proofs for typical coherence]{{Proofs for~\nameref{sec:typical-coherence}}}
\label{sec:proofs-coher}

The proof of~\autoref{prop:dist} requires the following lemma.

\begin{lemma}\label{lem:rows_big_max}
  Let $A \in \complex^{n \times k}$ be a matrix with $\ell_{2}$-normalized
  columns and let $A_i$ denote the $i$th row of $A$. Then
  \begin{align*}
    \max_{i \in [n]} \|A_i\|_2 \geq \sqrt{\frac{k}{n}}.
  \end{align*}
\end{lemma}

\begin{proof}[Proof of~\autoref{lem:rows_big_max}]
  Computing directly, using that each of the $k$ columns has unit norm,
  \begin{align*}
    \max_{i \in [n]} \|A_i\|_2^2 &\geq
    \operatornamewithlimits{mean}_{i \in [n]} \|A_i\|_2^2 
    = \frac{1}{n} \|A\|_F^2 = \frac{k}{n}.
  \end{align*}
  Taking square roots completes the proof.
\end{proof}

We now prove the proposition using the lemma.

\begin{proof}[{Proof of~\autoref{prop:dist}}]
  Take the set $\mathcal{T}$ of all subspaces of dimension $k$ in
  $\complex^n$. By rotational invariance of $\mathcal{T}$, it suffices to show
  the result with respect to $U = I$. Hence, let $\{e_i\}_{i \in [n]}$ be the
  canonical basis. Any fixed $T \in \mathcal{T}$ has coherence
  \begin{align*}
    \alpha = \sup_{v \in T \cap \sphn} \| v \|_\infty.
  \end{align*}
  We will show a sharp lower bound on the coherence of all $k$-dimensional
  subspaces, namely
  \begin{align*}
    \inf_{T \in \mathcal{T}} \sup_{v \in T \cap \sphn} \|v\|_\infty = \sqrt{\frac{k}{n}}.
  \end{align*}
  Take the set $\mathcal{A}\subseteq \complex^{n \times k}$ of orthonormal
  matrices. Since $\mathcal{T} = \{\range(A) : A \in \mathcal{A}\}$, it follows
  that
  \begin{align*}
    \inf_{T \in \mathcal{T}} \sup_{v \in T \cap \sphn} \|v\|_\infty   
    &= \inf_{A \in \mathcal{A}} \sup_{\nu \in \sph^{k-1}} \|A \nu\|_\infty
    \\
    &= \inf_{A \in \mathcal{A}} \max_{i \in [n]} \sup_{\nu \in \sph^{k-1}} |A_i \nu|
    \\
    &= \inf_{A \in \mathcal{A}} \max_{i \in [n]} \|A_i\|_2.
  \end{align*}
  Apply~\autoref{lem:rows_big_max} to lower bound the latter
  quantity. As~\autoref{lem:rows_big_max} applies to any matrix in
  $\mathcal{A}$,
  \begin{align*}
    \inf_{A \in \mathcal{A}} \sup_{i \in [n]} \|A_i\|_2 \geq \sqrt{\frac{k}{n}}.
  \end{align*}

  We next show that there exists a subspace such that equality holds. Take
  $F \in \complex^{n \times k}$ whose columns are the first $k$ columns of the
  DFT matrix, as defined in~\autoref{rmk:dft}. The columns of $F$ are
  orthonormal, so $F \in \mathcal{A}$. Furthermore, each row of $F$ has $\ell_2$
  norm $\sqrt{\frac{k}{n}}$. It follows that
  \begin{align*}
  \inf_{A \in \mathcal{A}} \max_{i \in [n]} \|A_i\|_2 = \sqrt{\frac{k}{n}}.
  \end{align*}
\end{proof}

The proof of~\autoref{thm:typical-coherence-gnn} uses~\autoref{lem:max_vars} and
the following lemma, which bounds the coherence of a random subspace sampled
from the \emph{Grassmannian}. The Grassmanian $\grassmannian_{n,k}$ consists of
all $k$-dimensional subspaces of $\reals^n$~\cite[Ch.~5.2.6]{vershynin2018high}.

\begin{lemma}
  \label{lem:coherence_rand_subspace}
  Let $U \in \complex^{n \times n}$ be a unitary matrix and denote by
  $\mathcal{L} \in \grassmannian_{n,k}$ a subspace distributed uniformly at
  random over $\grassmannian_{n,k}$. With probability at least
  $1 - 2\exp(-\gamma^2)$, $\mathcal{L}$ is $\alpha$-coherent with respect to
  $\|\cdot\|_U$ with
  \begin{align*}
  \alpha \lesssim \sqrt{\frac{k}{n}} + \sqrt \frac{\log(n)}{n} + \frac{\gamma}{\sqrt{ n }}.
  \end{align*}
\end{lemma}


\begin{proof}[Proof of~{~\autoref{thm:typical-coherence-gnn}}]
  Let 
  \begin{align*}
    \tilde{G}(x) = \sigma(W^{(d-1)} \cdots \sigma(W^{(1)}x))
  \end{align*}
  so that $G(x) = W^{(d)} \tilde{G}(x)$.
  Then,
  \begin{align*}
    \Delta(\range(G) - \range(G))
    &= \Delta \left(W^{(d)} (\range(\tilde G) - \range(\tilde G))\right)
    \\
    &= W^{(d)} \Delta \left( \range(\tilde G) - \range(\tilde G)\right).
  \end{align*}
  By~\autoref{lem:upper-bound-number-of-cones}
  and~\autoref{rmk:Delta-properties},
  $\Delta(\range(\tilde{G}) -\range(\tilde{G}))$ is the union of $M$ at-most
  $2k$-dimensional linear subspaces with
  \begin{align*}
    \log M \leq 2k \sum_{i = 1}^{d-1} \log \left( \frac{2ek_{i}}{k} \right).
  \end{align*}
  Each subspace $\mathcal{L}$ is uniformly distributed on
  $\grassmannian_{n,\dim\mathcal{L}}$, where $\dim\mathcal{L} \leq 2k$, because
  the final weight matrix has iid Gaussian entries independent of the other
  weight matrices (\eg see~\cite[Ch.~3.3.2]{vershynin2018high}). Enumerate the
  subspaces from $1$ to $M$ (independently of $W^{(d)}$). Then,
  applying~\autoref{lem:coherence_rand_subspace}, we see that the coherence of
  subspace $i$ is a random variable $\alpha_i$ such that
  \begin{align*}
    \alpha_i %
    \lesssim \sqrt{\frac{k}{n}} + \sqrt \frac{\log(n)}{n} + \frac{\gamma}{\sqrt{ n }}
  \end{align*}
  with probability at least $1-2\exp(- \gamma^2)$.
  
  Let $\alpha$ be the coherence of $\Delta(\range(G) - \range(G))$ and observe
  that $\alpha \leq \max_{i \in [M]} \alpha_i$. Applying~\autoref{lem:max_vars},
  \begin{align*}
    \alpha %
    \leq \max_{i \in [M]}\alpha_i \lesssim \sqrt{\frac{k}{n}} + \sqrt \frac{\log(n)}{n} + \sqrt{\frac{\log M}{n}} + \frac{\gamma}{\sqrt{n}}
  \end{align*}
  with probability at least $1-2\exp(-\gamma^2)$.
 \end{proof} 
  
We next prove~\autoref{lem:coherence_rand_subspace}.

\begin{proof}[Proof of~{~\autoref{lem:coherence_rand_subspace}}]
  Let $A \in \reals^{n\times k}$ with $A_{ij} \iid \mathcal{N}(0,1)$. Then,
  $\mathcal{L} := \range(A)$ is a random subspace uniformly distributed over
  $\grassmannian_{n,k}$. By rotation invariance of the Grassmannian, it suffices
  to show the result for $U = I$. Let $\{e_i\}_{i \in [n]}$ denote the canonical
  basis. Define
  \begin{align*}
    \alpha := \sup_{v \in \mathcal{L} \cap \sphn} \max_{i \in [n]} \left| \ip{e_i, v}\right|,
  \end{align*}
  and note that $\mathcal{L}$ is $\alpha$-coherent with respect to
  $\| \cdot \|_I = \| \cdot \|_\infty$. For each $i\in[n]$, we next analyze
  \begin{align}
    \label{eq:single_vec_coherence}
    \alpha_i %
    := \sup_{v \in \mathcal{L} \cap \sphn } |\langle e_i, v \rangle| %
    = \sup_{y \in \sph^{k-1}} \left|\frac{A_{i}y}{\|Ay\|_2}\right|.
  \end{align}
  We will show, with probability at least $1-4\exp(-s^2)$,
  \begin{align*}
    \alpha_i \lesssim \sqrt{ \frac{k}{n} }  + \frac{s}{\sqrt{ n }}.
  \end{align*}

  To see why this result should hold, we focus our attention on the right hand
  side of~\eqref{eq:single_vec_coherence}. The denominator concentrates around
  $\sqrt{ n }$ and the numerator is bounded by $\|A_{i}\|_2$, which concentrates
  around $\sqrt{ k }$ with subgaussian tails.

  We first obtain a lower bound on the smallest singular value of $A$
  via~\cite[Theorem~4.6.1]{vershynin2018high}, which guarantees with probability
  at least $1-2\exp(-t^2)$ that
  \begin{align*}
    \inf_{y \in \sph^{k-1}} \|Ay\|_2 \geq \sqrt{n} - C \sqrt{k} - Ct.
  \end{align*}
  By fixing $t=\frac{\sqrt{ n }}{2C}$ we define the event
  \begin{align*}
    B %
    := \left\{ \inf_{y \in \sph^{k-1}}\|Ay\|_2 \leq \frac{\sqrt{ n }}{2} - C\sqrt{ k } \right\}
  \end{align*}
  satisfying $\pr \left\{B\right\} \leq 2\exp(-c n)$. We first limit $s$ so that
  $\sqrt{ k }+s \leq \frac{\sqrt{ n }}{2} - C\sqrt{ k }$, which implies that
  $s < C\sqrt{ n }$. Then
  \begin{align*}
    & \pr\left\{\sup_{y \in \sph^{k-1}}  \frac{A_{i}y}{\|Ay\|_2} > \frac{\sqrt{ k }+s}{\frac{\sqrt{ n }}{2} - C\sqrt{ k }}\right\}	
    \\
    & = \pr\left\{\sup_{y \in \sph^{k-1}}  \frac{A_{i}y}{\|Ay\|_2} > \frac{\sqrt{ k }+s}{\frac{\sqrt{ n }}{2} - C\sqrt{ k }} \Big| B \right\}\pr\left\{B\right\}
    \\ 
    &+ \pr\left\{\sup_{y \in \sph^{k-1}}  \frac{A_{i}y}{\|Ay\|_2} > \frac{\sqrt{ k }+s}{\frac{\sqrt{ n }}{2} - C\sqrt{ k }} \Big| B^c\right\}\pr \left\{B^c \right\}
    \\
    & \leq\pr\{B\} + \pr\left\{\sup_{y \in \sph^{k-1}}  A_{i}y  > \sqrt{ k } + s \right\}
    \\
    & \leq 2\exp(-cn) + \pr\left\{\|A_{i}\|_{2}  > \sqrt{ k } + s \right\}
    \\
    & \leq 2\exp(-cn) + 2\exp(-cs^2).
  \end{align*}
  Above, we used the concentration of the norm of Gaussian
  vectors~\cite[Theorem~3.1.1]{vershynin2018high}. Since $s \leq C \sqrt{ n }$,
  $\exp(-cn) \leq \exp(-cs^2)$. From this we find the desired subgaussian tail
  bound:
  \begin{align*}
    \pr\left\{\sup_{y \in \sph^{k-1}}  %
    \frac{A_{i}y}{\|Ay\|_2} %
    > \frac{\sqrt{ k }+s}{\frac{\sqrt{ n }}{2} - C\sqrt{ k }}\right\}	
    \leq 4\exp(-cs^2).
  \end{align*}
  The remaining values of $s$ satisfy
  $\sqrt{ k }+s > \frac{\sqrt{ n }}{2} - C\sqrt{ k }$. Therefore, since
  $\sup_{y \in \sph^{k-1}}\frac{A_{i}y}{\|Ay\|_2} \leq 1$,
  \begin{align*}
    \pr\left\{%
    \sup_{y \in \sph^{k-1}}  \frac{A_{i}y}{\|Ay\|_2} %
    > \frac{\sqrt{ k }+s}{\frac{\sqrt{ n }}{2} - C\sqrt{ k }}\right\}	%
    = 0 %
    \leq 4\exp(-cs^2).
  \end{align*}
  Therefore, the subgaussian bound applies for all values of $s$.

  We now scale $s$ by an absolute constant with $\gamma = c s$. Then
  \begin{align*}
    \alpha_i %
    = \sup_{y \in \sph^{k-1}}  \left |\frac{A_{i}y}{\|Ay\|_2}\right | %
    > \frac{\sqrt{ k }+s}{\frac{\sqrt{ n }}{2} - C\sqrt{ k }} %
    \gtrsim \sqrt{\frac{k}{n}} + \frac{\gamma}{\sqrt{n}}
  \end{align*}
  with probability less than $2\exp(-\gamma^2)$. Changing the constant from 4 to
  2 in the probability bound is achieved by suitable choice of $c$. Remembering
  that $\alpha$ is the coherence with the canonical basis, we
  apply~\autoref{lem:max_vars} to find,
  \begin{align*}
    \alpha %
    = \max_{i \in [n]} \alpha_i %
    \lesssim \sqrt{\frac{k}{n}}+\sqrt{\frac{\log n}{n}} + \frac{\gamma}{\sqrt{n}}
  \end{align*}
  with probability at least $1 - 2\exp(-\gamma^2)$.
\end{proof}

\section{Conclusion}
\label{sec:conclusion}

In this work, we have proved a restricted isometry property for a subsampled
isometry with GNN structural proxy, \autoref{thm:G-RIP}. We used this to prove
sample complexity and recovery bounds, \autoref{thm:subIso-GCS}. The recovery
bound stated in~\autoref{thm:subIso-GCS} is uniform over ground truth signals,
and permits a more finely tuned nonuniform control as discussed
in~\autoref{rmk:modelling-error}. To our knowledge, this provides the first
theory for generative compressed sensing with subsampled isometries and
non-random weights.

Our results rely on the notion of $\alpha$-coherence with respect to the
measurement norm, introduced in~\autoref{def:coherence}
and~\autoref{def:measurement-norm}, respectively. Closely related to the notion
of incoherent bases~\cite[p.~373]{foucart2017mathematical} and the $X$-norm
of~\cite{rudelson2008sparse}, we argue that $\alpha$-coherence is a natural
quantity to measure the interplay between a GNN and the measurement
operator. Indeed, in~\autoref{sec:numerics} we propose a regularization strategy
for promoting favourable coherence of GNNs during training, and connect this
strategy with favourable recovery efficacy. Specifically, we show that our
regularization strategy yields low coherence GNNs with improved sample
complexity for recovery~(\autoref{fig:MNIST_random_comp}). Moreover, our
numerics support that low coherence GNNs achieve better sample complexity than
high coherence GNNs~(\autoref{fig:MNIST-recovery-dependence}).

We suspect the $\Omega(k^2d^2)$ dependence in the sample complexity of our
analysis is sub-optimal, and a consequence of our coherence-based approach.
Ignoring logarithmic factors, it is an open question to prove recovery
guarantees with $\Omega(kd)$ Fourier measurements and non-random weights, which
would match the number of (sub-)Gaussian measurements
needed~\cite{bora2017compressed} and would also match the known worst-case lower
bound~\cite{liu2020information}. In addition, it is an open problem to improve
the regularization strategy for lowering coherence, possibly including middle
layers. Finally, it is open to determine a notion of coherence for networks that
have a final nonlinear activation function, and to characterize how this impacts
recovery efficacy for such networks.

\section*{Acknowledgement}

The authors would like to thank Ben Adcock for finding an error in an early
version of this manuscript.

\bibliographystyle{ieeetr}
\bibliography{gcs-subIso}

\begin{thebibliography}{10}

\bibitem{bora2017compressed}
A.~Bora, A.~Jalal, E.~Price, and A.~G. Dimakis, ``Compressed sensing using
  generative models,'' in {\em International Conference on Machine Learning},
  pp.~537--546, 2017.

\bibitem{scarlett2022theoretical}
J.~Scarlett, R.~Heckel, M.~R. Rodrigues, P.~Hand, and Y.~C. Eldar,
  ``Theoretical perspectives on deep learning methods in inverse problems,''
  {\em arXiv preprint arXiv:2206.14373}, 2022.

\bibitem{kumar2015source}
R.~Kumar, H.~Wason, and F.~J. Herrmann, ``Source separation for simultaneous
  towed-streamer marine acquisition - a compressed sensing approach,'' {\em
  Geophysics}, vol.~80, no.~6, pp.~WD73--WD88, 2015.

\bibitem{herrmann2012fighting}
F.~J. Herrmann, M.~P. Friedlander, and {\"O}.~Yilmaz, ``Fighting the curse of
  dimensionality: Compressive sensing in exploration seismology,'' {\em IEEE
  Signal Processing Magazine}, vol.~29, no.~3, pp.~88--100, 2012.

\bibitem{adcock2021compressive}
B.~Adcock and A.~C. Hansen, {\em Compressive Imaging: Structure, Sampling,
  Learning}.
\newblock Cambridge University Press, Cambridge, UK, 2021.

\bibitem{lustig2008compressed}
M.~Lustig, D.~L. Donoho, J.~M. Santos, and J.~M. Pauly, ``Compressed sensing
  mri,'' {\em IEEE Signal Processing Magazine}, vol.~25, no.~2, pp.~72--82,
  2008.

\bibitem{foucart2017mathematical}
S.~Foucart and H.~Rauhut, {\em A Mathematical Introduction to Compressive
  Sensing}.
\newblock Applied and Numerical Harmonic Analysis, Birkhäuser, New York, NY,
  2013.

\bibitem{jalal2021robust}
A.~Jalal, M.~Arvinte, G.~Daras, E.~Price, A.~G. Dimakis, and J.~Tamir, ``Robust
  compressed sensing {MRI} with deep generative priors,'' {\em Advances in
  Neural Information Processing Systems}, vol.~34, pp.~14938--14954, 2021.

\bibitem{kingma2013auto}
D.~P. Kingma and M.~Welling, ``Auto-encoding variational {Bayes},'' {\em arXiv
  preprint arXiv:1312.6114}, 2013.

\bibitem{goodfellow2014generative}
I.~Goodfellow, J.~Pouget-Abadie, M.~Mirza, B.~Xu, D.~Warde-Farley, S.~Ozair,
  A.~Courville, and Y.~Bengio, ``Generative adversarial nets,'' in {\em
  Advances in Neural Information Processing Systems}, pp.~2672--2680, 2014.

\bibitem{radford2015unsupervised}
A.~Radford, L.~Metz, and S.~Chintala, ``Unsupervised representation learning
  with deep convolutional generative adversarial networks,'' {\em arXiv
  preprint arXiv:1511.06434}, 2015.

\bibitem{dirksen2015tail}
S.~Dirksen, ``Tail bounds via generic chaining,'' {\em Electronic Journal of
  Probability}, vol.~20, 2015.

\bibitem{rudelson2008sparse}
M.~Rudelson and R.~Vershynin, ``On sparse reconstruction from {Fourier} and
  {Gaussian} measurements,'' {\em Communications on Pure and Applied
  Mathematics: A Journal Issued by the Courant Institute of Mathematical
  Sciences}, vol.~61, no.~8, pp.~1025--1045, 2008.

\bibitem{donoho2003optimally}
D.~L. Donoho and M.~Elad, ``Optimally sparse representation in general
  (nonorthogonal) dictionaries via $\ell^1$ minimization,'' {\em Proceedings of
  the National Academy of Sciences}, vol.~100, no.~5, pp.~2197--2202, 2003.

\bibitem{candes2011probabilistic}
E.~J. Cand{\`e}s and Y.~Plan, ``A probabilistic and {RIPless} theory of
  compressed sensing,'' {\em IEEE Transactions on Information Theory}, vol.~57,
  no.~11, pp.~7235--7254, 2011.

\bibitem{berk2021deep}
A.~Berk, ``Deep generative demixing: Error bounds for demixing subgaussian
  mixtures of {Lipschitz} signals,'' in {\em ICASSP 2021-2021 IEEE
  International Conference on Acoustics, Speech and Signal Processing
  (ICASSP)}, pp.~4010--4014, IEEE, 2021.

\bibitem{liu2022non}
J.~Liu and Z.~Liu, ``Non-iterative recovery from nonlinear observations using
  generative models,'' {\em arXiv preprint arXiv:2205.15749}, 2022.

\bibitem{hand2019global}
P.~Hand and V.~Voroninski, ``Global guarantees for enforcing deep generative
  priors by empirical risk,'' {\em IEEE Transactions on Information Theory},
  vol.~66, no.~1, pp.~401--418, 2019.

\bibitem{joshi2021plugin}
B.~Joshi, X.~Li, Y.~Plan, and {\"O}.~Yilmaz, ``{PLUGIn}: A simple algorithm for
  inverting generative models with recovery guarantees,'' {\em Advances in
  Neural Information Processing Systems}, vol.~34, 2021.

\bibitem{joshi2019global}
P.~Hand and B.~Joshi, ``Global guarantees for blind demodulation with
  generative priors,'' in {\em Advances in Neural Information Processing
  Systems}, vol.~32, pp.~11535--11545, 2019.

\bibitem{hand2018phase}
P.~Hand, O.~Leong, and V.~Voroninski, ``Phase retrieval under a generative
  prior,'' {\em Advances in Neural Information Processing Systems}, vol.~31,
  2018.

\bibitem{deora2020structure}
P.~Deora, B.~Vasudeva, S.~Bhattacharya, and P.~M. Pradhan, ``Structure
  preserving compressive sensing {MRI} reconstruction using generative
  adversarial networks,'' in {\em Proceedings of the IEEE/CVF Conference on
  Computer Vision and Pattern Recognition Workshops}, pp.~522--523, 2020.

\bibitem{mardani2018deep}
M.~Mardani, E.~Gong, J.~Y. Cheng, S.~S. Vasanawala, G.~Zaharchuk, L.~Xing, and
  J.~M. Pauly, ``Deep generative adversarial neural networks for compressive
  sensing {MRI},'' {\em IEEE Transactions on Medical Imaging}, vol.~38, no.~1,
  pp.~167--179, 2018.

\bibitem{li2022fast}
W.~Li, A.~Zhu, Y.~Xu, H.~Yin, and G.~Hua, ``A fast multi-scale generative
  adversarial network for image compressed sensing,'' {\em Entropy}, vol.~24,
  no.~6, p.~775, 2022.

\bibitem{wentz2022genmod}
J.~Wentz and A.~Doostan, ``{GenMod}: A generative modeling approach for
  spectral representation of {PDEs} with random inputs,'' {\em arXiv preprint
  arXiv:2201.12973}, 2022.

\bibitem{ulyanov2018deep}
D.~Ulyanov, A.~Vedaldi, and V.~Lempitsky, ``Deep image prior,'' in {\em
  Proceedings of the IEEE Conference on Computer Vision and Pattern
  Recognition}, pp.~9446--9454, 2018.

\bibitem{heckel2019deep}
R.~Heckel and P.~Hand, ``{Deep Decoder}: Concise image representations from
  untrained non-convolutional networks,'' in {\em International Conference on
  Learning Representations}, 2019.

\bibitem{darestani2021accelerated}
M.~Z. Darestani and R.~Heckel, ``Accelerated {MRI} with un-trained neural
  networks,'' {\em IEEE Transactions on Computational Imaging}, vol.~7,
  pp.~724--733, 2021.

\bibitem{candes2006robust}
E.~J. Cand{\`e}s, J.~Romberg, and T.~Tao, ``Robust uncertainty principles:
  Exact signal reconstruction from highly incomplete frequency information,''
  {\em IEEE Transactions on Information Theory}, vol.~52, no.~2, pp.~489--509,
  2006.

\bibitem{donoho2006compressed}
D.~L. Donoho, ``Compressed sensing,'' {\em IEEE Transactions on Information
  Theory}, vol.~52, no.~4, pp.~1289--1306, 2006.

\bibitem{rauhut2010compressive}
H.~Rauhut, ``Compressive sensing and structured random matrices,'' in {\em
  Theoretical foundations and numerical methods for sparse recovery},
  pp.~1--92, de Gruyter, 2010.

\bibitem{bourgain2014improved}
J.~Bourgain, ``An improved estimate in the restricted isometry problem,'' in
  {\em Geometric aspects of functional analysis}, pp.~65--70, Springer, 2014.

\bibitem{haviv2017restricted}
I.~Haviv and O.~Regev, ``The restricted isometry property of subsampled
  {Fourier} matrices,'' in {\em Geometric aspects of functional analysis},
  pp.~163--179, Springer, 2017.

\bibitem{chkifa2018polynomial}
A.~Chkifa, N.~Dexter, H.~Tran, and C.~G. Webster, ``Polynomial approximation
  via compressed sensing of high-dimensional functions on lower sets,'' {\em
  Mathematics of Computation}, vol.~87, no.~311, pp.~1415--1450, 2018.

\bibitem{brugiapaglia2021sparse}
S.~Brugiapaglia, S.~Dirksen, H.~C. Jung, and H.~Rauhut, ``Sparse recovery in
  bounded {R}iesz systems with applications to numerical methods for {PDEs},''
  {\em Applied and Computational Harmonic Analysis}, vol.~53, pp.~231--269,
  2021.

\bibitem{naderiPlanSparsityFree}
A.~Naderi and Y.~Plan, ``Sparsity-free compressed sensing with generative
  priors as special case,'' 2022.
\newblock Unpublished manuscript.

\bibitem{candes2007sparsity}
E.~J. Cand{\`e}s and J.~Romberg, ``Sparsity and incoherence in compressive
  sampling,'' {\em Inverse problems}, vol.~23, no.~3, p.~969, 2007.

\bibitem{cape2019two}
J.~Cape, M.~Tang, and C.~E. Priebe, ``The two-to-infinity norm and singular
  subspace geometry with applications to high-dimensional statistics,'' {\em
  The Annals of Statistics}, vol.~47, no.~5, pp.~2405--2439, 2019.

\bibitem{liaw2017simple}
C.~Liaw, A.~Mehrabian, Y.~Plan, and R.~Vershynin, ``A simple tool for bounding
  the deviation of random matrices on geometric sets,'' in {\em Geometric
  aspects of functional analysis}, pp.~277--299, Springer, 2017.

\bibitem{jeong2020sub}
H.~Jeong, X.~Li, Y.~Plan, and {\"O}.~Yilmaz, ``Sub-gaussian matrices on sets:
  Optimal tail dependence and applications,'' {\em arXiv preprint
  arXiv:2001.10631}, 2020.

\bibitem{vershynin2018high}
R.~Vershynin, {\em High-dimensional Probability: An Introduction with
  Applications in Data Science}.
\newblock Cambridge University Press, Cambridge, UK, 2018.

\bibitem{naderi2021beyond}
A.~Naderi and Y.~Plan, ``Beyond independent measurements: General compressed
  sensing with {GNN} application,'' {\em arXiv preprint arXiv:2111.00327},
  2021.

\bibitem{serra2018bounding}
T.~Serra, C.~Tjandraatmadja, and S.~Ramalingam, ``Bounding and counting linear
  regions of deep neural networks,'' in {\em International Conference on
  Machine Learning}, pp.~4558--4566, PMLR, 2018.

\bibitem{novak2018sensitivity}
R.~Novak, Y.~Bahri, D.~A. Abolafia, J.~Pennington, and J.~Sohl-Dickstein,
  ``Sensitivity and generalization in neural networks: An empirical study,''
  {\em arXiv preprint arXiv:1802.08760}, 2018.

\bibitem{deng2012mnist}
L.~Deng, ``The {MNIST} database of handwritten digit images for machine
  learning research,'' {\em IEEE Signal Processing Magazine}, vol.~29, no.~6,
  pp.~141--142, 2012.

\bibitem{kingma2014adam}
D.~P. Kingma and J.~Ba, ``Adam: A method for stochastic optimization,'' {\em
  arXiv preprint arXiv:1412.6980}, 2014.

\bibitem{innes2018flux}
M.~Innes, ``Flux: Elegant machine learning with julia,'' {\em Journal of Open
  Source Software}, vol.~3, no.~25, p.~602, 2018.

\bibitem{code-repo}
A.~Berk, S.~Brugiapaglia, B.~Joshi, Y.~Plan, M.~Scott, and O.~Yilmaz, ``{subIso
  GCS},'' {\em {GitHub repository}}, 2022.
\newblock \url{https://github.com/babhrujoshi/GNN-with-sub-Fourier-paper}.

\bibitem{kingma2019introduction}
D.~P. Kingma and M.~Welling, ``An introduction to variational autoencoders,''
  {\em Foundations and Trends in Machine Learning}, vol.~12, no.~4,
  pp.~307--392, 2019.

\bibitem{virtanen2020scipy}
P.~Virtanen, R.~Gommers, T.~E. Oliphant, M.~Haberland, T.~Reddy, D.~Cournapeau,
  E.~Burovski, P.~Peterson, W.~Weckesser, J.~Bright, S.~J. {van der Walt},
  M.~Brett, J.~Wilson, K.~J. Millman, N.~Mayorov, A.~R.~J. Nelson, E.~Jones,
  R.~Kern, E.~Larson, C.~J. Carey, {\.I}.~Polat, Y.~Feng, E.~W. Moore,
  J.~{VanderPlas}, D.~Laxalde, J.~Perktold, R.~Cimrman, I.~Henriksen, E.~A.
  Quintero, C.~R. Harris, A.~M. Archibald, A.~H. Ribeiro, F.~Pedregosa, P.~{van
  Mulbregt}, and {SciPy 1.0 Contributors}, ``{SciPy 1.0}: Fundamental
  algorithms for scientific computing in {Python},'' {\em Nature Methods},
  vol.~17, no.~3, pp.~261--272, 2020.

\bibitem{adcock2022sparse}
B.~Adcock, S.~Brugiapaglia, and C.~G. Webster, {\em Sparse Polynomial
  Approximation of High-Dimensional Functions}.
\newblock SIAM, Philadelphia, PA, 2022.

\bibitem{liu2020information}
Z.~Liu and J.~Scarlett, ``Information-theoretic lower bounds for compressive
  sensing with generative models,'' {\em IEEE Journal on Selected Areas in
  Information Theory}, 2020.

\bibitem{tropp2012user}
J.~A. Tropp, ``User-friendly tail bounds for sums of random matrices,'' {\em
  Foundations of Computational Mathematics}, vol.~12, no.~4, pp.~389--434,
  2012.

\bibitem{boucheron2013concentration}
S.~Boucheron, G.~Lugosi, and P.~Massart, {\em Concentration Inequalities: A
  Nonasymptotic Theory of Independence}.
\newblock Oxford University Press, Oxford, UK, 2013.

\bibitem{jacques2013robust}
L.~Jacques, J.~N. Laska, P.~T. Boufounos, and R.~G. Baraniuk, ``Robust 1-bit
  compressive sensing via binary stable embeddings of sparse vectors,'' {\em
  IEEE Transactions on Information Theory}, vol.~59, no.~4, pp.~2082--2102,
  2013.

\bibitem{cover1965geometrical}
T.~M. Cover, ``Geometrical and statistical properties of systems of linear
  inequalities with applications in pattern recognition,'' {\em IEEE
  Transactions on Electronic Computers}, vol.~EC-14, no.~3, pp.~326--334, 1965.

\bibitem{flatto1970new}
L.~Flatto, ``A new proof of the transposition theorem,'' {\em Proceedings of
  the American Mathematical Society}, vol.~24, no.~1, pp.~29--31, 1970.

\bibitem{bauschke2011convex}
H.~H. Bauschke and P.~L. Combettes, {\em Convex analysis and monotone operator
  theory in Hilbert spaces}, vol.~408.
\newblock Springer, 2011.

\end{thebibliography}

\clearpage
\begin{center}
  {\Large\textbf{Supplementary material}}
\end{center}
\setcounter{section}{0}

\makeatletter{}
\renewcommand\thesection{S\@arabic\c@section}
\makeatother{}

This supplementary material contains auxiliary results from high-dimensional
probability (\autoref{sec:results-from-hdp}); a characterization of the range of
a generative network and an elaboration on properties of the operator $\Delta$
(\autoref{sec:characterize-range-G} and~\autoref{sec:cvxcone-and-Delta}
respectively); a discussion concerning control of the approximation error
(\autoref{sec:control-approx-error}) as it pertains to
\autoref{rmk:modelling-error}; and an epilogue (\autoref{sec:epilogue}), whose
purpose is to clarify the role of our mathematical argument in its determination
of the sample complexity.

\section{Results from high-dimensional probability}
\label{sec:results-from-hdp}

\subsection{Matrix Bernstein inequality} 

For this result, see~\cite[Theorem~6.1]{tropp2012user}.

\begin{lemma}[Matrix Bernstein inequality]
  \label{lem:matrix-bernstein}
  Let $X_{1}, \ldots, X_{N} \in \complex^{n\times n}$ be independent, mean-zero,
  self-adjoint random matrices, such that $\|X_{i}\| \leq K$ almost surely for
  all $i$. Then, for every $\gamma \geq 0$ we have
  \begin{align*}
    \pr \left\{ \left\| \sum_{i = 1}^{N} X_{i} \right\| \geq \gamma \right\} %
    & \leq 2n \exp \left( - \frac{\gamma^{2}/2}{\tau^{2} + K\gamma /3} \right)
    \\
    & \leq 2n \exp \left( - c \cdot \min \left( \frac{\gamma^{2}}{\tau^{2}}, \frac{\gamma}{K} \right) \right).
  \end{align*}
  Here, $\tau^{2} := \left\| \sum_{i = 1}^{N} \E X_{i}^2 \right\|$. 
\end{lemma}

\subsection{Cram\'er-Chernoff bound}
\label{sec:cram-chern-bound}

For a reference to this material,
see~\cite[Ch.~2.2]{boucheron2013concentration}. Let $Z$ be a real-valued random
variable. Then,
\begin{align*}
  \pr\left\{ Z \geq t \right\} %
  &\leq \exp ( - \psi_{Z}^{*} (t)),
  & %
    \psi_{Z}^{*}(t) %
  &:= \sup_{ \lambda \geq 0} \lambda t - \psi_{Z}(\lambda).
\end{align*}
The latter quantity, $\psi_{Z}^{*}(t)$, is the Cram\'er transform of $Z$, with
\begin{align*}
  \psi_{Z}(\lambda) := \log \E \exp(\lambda Z), \quad \lambda \geq 0
\end{align*}
being the logarithm of the moment generating function of $Z$. When $Z$ is
centered, $\psi_Z$ is continuously differentiable on an interval of the form
$[0,b)$ and $\psi_Z(0) = \psi'_Z(0) = 0$. Thus,
\begin{align*}
\psi^*_Z(t) = \lambda_t t - \psi_Z(\lambda_t),
\end{align*}
where $\lambda_t$ is such that $\psi'_Z(\lambda_t) = t$.

For a centered binomial random variable $Z := Y - np$ where
$Y \sim \DistBinom(n, p)$, the Cram\'er transform of $Z$ is given by
\begin{align*}
  \psi_{Z}^{*}(t) %
  &:= n h_{p} \left(  \frac{t}{n} + p \right), \quad \forall 0 < t < n ( 1- p), 
\end{align*}
where
\begin{align*}
  h_{p}(a) := (1 - a) \log \frac{1-a}{1-p} + a \log \frac{a}{p}
\end{align*}
is the Kullback-Leibler divergence $\kldiv{P_{a}}{P_{p}}$ between Bernoulli
distributions with parameters $a$ and $p$. One may thus
establish~\cite[Ch.~2.2]{boucheron2013concentration} the following concentration
inequality for $Y$ when $0 < t < \frac{n}{m}$:
\begin{align*}
  \pr\left\{ Y \geq tm \right\} %
  = \pr\left\{ Z \geq (t - 1)m \right\} %
  \leq \exp(-n h_{p}(tp)).
\end{align*}

\subsection{Auxiliary union bound}
\label{sec:aux-union-bound}

\begin{lemma}
  \label{lem:max_vars}
  Let $f:\reals\to \reals$ be an increasing function. Let $\{X_i\}_{i \in [n]}$
  be a collection of random variables such that for each $i \in [n]$ and for any
  $\gamma > 0$,
  \begin{align*}
    X_{i} \leq f(\gamma)
  \end{align*}
  with probability at least $1 - 2\exp(-\gamma^2)$. Then, for all $\gamma > 0$,
  \begin{align*}
    \max_{i \in [n]} X_{i}
    \leq f\left(c\left(\gamma + \sqrt{ \log n }\right)\right)
  \end{align*}
  with probability at least $1 - 2\exp(-\gamma^2)$.
\end{lemma}

\begin{proof}[{Proof of~\autoref{lem:max_vars}}]
  By assumption,
  \begin{align*}
    \pr\left\{\max_{i \leq [n]} X_{i} > f(\gamma)\right\}
    &= \pr\left\{ \bigcup_{i \in [n]} \{X_i > f(\gamma) \}\right\}
    \\
    &\leq n 2  \exp(-\gamma^2) 
    \\
    &= 2\exp(-\gamma^2 + \log n).
  \end{align*}
  Let $t := \sqrt{\gamma^2 - \log n}$. We substitute $\gamma^2 \to t^2 + \log n$
  in the right hand side of the equation above which yields
  \begin{align*}
    \pr\left\{\max_{i \leq [n]} X_{i} > f(\gamma)\right\}
    \leq 2\exp(-t^2).
  \end{align*}
  To substitute the remaining $\gamma$ on the left hand side, first notice
  \begin{align*} 
    \gamma &\leq 2\max \left(t, \sqrt{\log n}\right) \leq 2 \left(t + \sqrt{\log n}\right).
  \end{align*}
  Consequently, $f(\gamma) \leq f\left(2\left(t + \sqrt{\log n}\right)\right)$ and
  \begin{align*}
    \pr \left\{ \max_{i \in [n]} X_i > f\left(2\left(t + \sqrt{\log n}\right)\right) \right\} 
    & \leq \pr \left\{ \max_{i \in [n]} X_i > f(\gamma) \right\} 
    \\
    & \leq 2 \exp(-t^2)
  \end{align*}
  Re-labelling $t \to \gamma$ yields the result.
\end{proof}

\section{\texorpdfstring{Characterizing $\range(G)$}{Characterizing range(G)}}
\label{sec:characterize-range-G}

For completeness, we include a characterization of the geometry of the range of
generative neural networks with ReLU activation. This material is not novel; for
instance, see~\cite{naderi2021beyond}. It requires the notion of a
low-dimensional cone.

\begin{definition}
  A convex set $\mathcal{C} \subseteq \reals^{n}$ is \emph{at-most
    $k$-dimensional} for some $k \in [n]$ if there exists a linear subspace
  $E \subseteq \reals^{n}$ with $\dim (E) \leq k$ and $\mathcal{C} \subseteq E$.
\end{definition}

\begin{remark}
  A cone $\mathcal{C} \subseteq \reals^{n}$ is at-most $k$-dimensional if its
  linear hull, $\Span(\mathcal{C})$, is no greater than $k$-dimensional.
\end{remark}

A key idea in our recovery guarantees will be to cover $\range(G)$ using at-most
$k$-dimensional cones. To this end, it will be important to bound the number $N$
of at-most $k$-dimensional cones in the covering set. This comes down to
quantifying the number of different ways that ReLU can act on a given subspace,
which reduces to counting the number of orthants that subspace intersects. The
following result may be found in~\cite[App.~A~Lemma~1]{jacques2013robust}.

\begin{lemma}[orthant-crossings]
  \label{lem:orthant-crossings}
  Let $S \subseteq \reals^{n}$ be a $k$-dimensional subspace. Then $S$
  intersects at most $I(n, k)$ different orthants where
  \begin{align}
    \label{eq:orthants}
    I(n, k) \leq 2^{k} \binom{n}{k} \leq 2^{k} \left(  \frac{en}{k}\right)^{k}. 
  \end{align}
\end{lemma}

Note that previous work~\cite{cover1965geometrical, flatto1970new} has
established the bound
\begin{align*}
  I(n, k) \leq 2 \sum_{\ell = 0}^{k-1} \binom{n-1}{\ell}, 
\end{align*}
which is tight~\cite{cover1965geometrical}. Accordingly, the extent of
non-tightness of the bound~\eqref{eq:orthants} may be evaluated directly for
specific choices of $n$ and $k$. The upper bound on $N$ is an immediate
consequence. We also require the notion of a polyhedral cone, used
in~\autoref{lem:upper-bound-number-of-cones} below. Note that a version of the
lemma may be found in~\cite{naderi2021beyond}.

\begin{definition}[Polyhedral cone]
  \label{def:polyhedral-cone}
  A cone $\mathcal{C} \subseteq \reals^n$ is called polyhedral if it is the
  conic combination of finitely many vectors. Equivalently, $\mathcal{C}$ is
  polyhedral if it is the intersection of a finite number of half-spaces that
  have the origin on their boundary.
\end{definition}

\begin{lemma}
  \label{lem:upper-bound-number-of-cones}
  Let $G$ be a $(k,d,n)$-generative network with layer widths
  $k=k_0\leq k_1, \ldots, k_d$ where $k_d = n$. Then $\range(G)$ is a union of
  no more than $N$ at-most $k$-dimensional polyhedral cones where
  \begin{align*}
    N %
    &:= \prod_{\ell = 1}^{d-1} I(k_{\ell}, k) %
    \leq \left( \frac{2e \bar k}{k}  \right)^{k(d-1)}, %
    \\
    \bar k &:= \left( \prod_{\ell = 1}^{d-1} k_{\ell}\right)^{1/(d-1)}. 
  \end{align*}
\end{lemma}

\begin{proof}[Proof of~{\autoref{lem:upper-bound-number-of-cones}}]
  The proof is similar to the one given in~\cite{naderi2021beyond}, repeated
  here for completeness.

  First, if $\mathcal{C} \subseteq \reals^{\tilde n}$ is a polyhedral cone of
  dimension $k \leq \tilde n$ and $L : \reals^{\tilde n} \to \reals^{n'}$ is a
  linear map with $n' \in \nats$ then $L\mathcal{C}$ is a polyhedral cone with
  dimension at most $\min\{k, n'\}$. Likewise, the intersection of a collection
  of polyhedral cones is a polyhedral cone. Now, if
  $(Q_{i})_{i \in [2^{\tilde n}]}$ are the orthants of $\reals^{\tilde n}$ with
  $Q_{1}$ being the nonnegative orthant, observe that $\proj_{Q_{1}}$ is linear
  on each orthant $Q_i$, $i \in [2^{\tilde n}]$. Consequently,
  \begin{align*}
    \sigma(\mathcal{C}) %
    = \bigcup\limits_{i=1}^{2^{\tilde n}} \proj_{Q_{1}} \left( \mathcal{C} \cap Q_{i} \right)
  \end{align*}
  is a union of polyhedral cones. Since the domain $\reals^{k}$ of $G$ is a
  polyhedral cone, it follows that $\sigma\left(W^{(1)} \reals^k\right)$ is a
  union of polyhedral cones, and continuing by induction, $\range(G)$ is a union
  of polyhedral cones. By the rank-nullity theorem each component cone has
  dimension at most $k$.

  We next argue for the bound $N$ on the number of cones comprising the
  range. Consider that the final mapping $W^{(d)}$ cannot increase the number of
  cones. Thus, it suffices to consider the map
  $\sigma(W^{(d-1)}\cdots \sigma(W^{(1)} u))$. We need only count the number of
  new intersected subspaces that could be generated from each ReLU operation. As
  there are $d-1$ in total, we obtain
  \begin{align*}
    \prod_{\ell = 1}^{d-1} I(k_{\ell}, k) %
    \leq 2^{(d-1) k} \prod_{\ell = 1}^{d-1}  \left( \frac{e k_{\ell}}{k} \right)^{k} %
    = \left( \frac{2e \bar k}{k}  \right)^{k(d-1)}, %
  \end{align*}
  where $\bar k$ is the geometric mean
  $\bar k := \left( \prod_{\ell = 1}^{d-1} k_{\ell}\right)^{1/(d-1)}$.%
\end{proof}

\begin{remark}
  \label{rmk:log-subspace-number-expression}
  A simple calculation shows that
  \begin{align*}
    \log N \leq k \sum_{i = 1}^{d-1} \log \left( \frac{2ek_{i}}{k} \right). 
  \end{align*}
  \vskip-12pt
\end{remark}

\begin{remark}
  \label{rmk:extend-to-range-difference}
  If $G$ is a $d$ layer neural network with ReLU activation then $G(x) - G(y)$
  can be written as $\bar G(x, y) := G(x) - G(y)$ where
  $\bar G : \reals^{2k} \to \reals^{n}$ is a $d$-layer neural network whose
  $i$th weight matrix is
  \begin{align*}
    \bmat{ll}{%
    W^{(i)} & \mathbf{0}_{k_{i} \times k_{i-1}} \\
    \mathbf{0}_{k_{i} \times k_{i-1}}  & W^{(i)}
    }
  \end{align*}
  for $i = 1, \ldots, d-1$ where $\mathbf{0}_{m, n}$ is the $m\times n$ zero
  matrix. The $d$th weight matrix of $\bar G$ is
  $\bmat{rr}{W^{(d)} & -W^{(d)}}$. In particular, the layer widths of $\bar G$
  are $2k_{0} = 2k, 2k_{1}, \ldots, 2k_{d-1}, k_{d} = n$, so
  applying~\autoref{lem:upper-bound-number-of-cones} to $\bar G$ gives
  \begin{align*}
    N(\bar G) %
    &\leq \left( \frac{2e\bar k}{k} \right)^{2k(d-1)}, %
    & %
      \log N(\bar G) %
    & \leq 2k \sum_{i=1}^{d-1} \log \left( \frac{2ek_{i}}{k} \right),
  \end{align*}
  where $\bar k$ is the geometric mean of the layer widths of \emph{the original
    network} $G$.
\end{remark}

\begin{remark}
  \label{rmk:typical-coherence-random-weights-with-bias}
  We argue that~\autoref{thm:typical-coherence-gnn} extends to GNNs with
  biases. That the result holds for GNNs with arbitrary fixed biases in all but
  the last layer follows from the representation of a biased GNN with augmented
  matrices, see~\autoref{rmk:networks-with-biases}. Allowing for an arbitrary
  fixed (or random) bias in the last layer then follows, because
  $\range(G)- \range(G)$ is invariant to changes in the bias of the last layer,
  since it is affine.
\end{remark}

\section{\texorpdfstring{Convex cones and the operator $\Delta$}{Convex cones
    and the operator Delta}}
\label{sec:cvxcone-and-Delta}

For the following proposition, see~\cite[Proposition~6.4]{bauschke2011convex}.

\begin{proposition}[Cone difference is subspace]
  \label{prop:cone-difference-is-subspace}
  Let $\mathcal{K} \subseteq \reals^{n}$ be a convex cone. Then
  $\mathcal{K} - \mathcal{K} = \Span \mathcal{K}$.
\end{proposition}

\begin{remark}
  \label{rmk:Delta-properties}
  Below we list several properties about $\Delta$. Let
  $\mathcal{C} = \bigcup\limits_{i=1}^N\mathcal{C}_i$ be the union of
  $N\in\nats$ convex cones $\mathcal{C}_{i}$.
  \begin{enumerate}[leftmargin=12pt]
  \item The latter equality in~\autoref{def:minimal-subspace-covering} follows
    from~\autoref{prop:cone-difference-is-subspace}.
  \item The set $\Delta(\mathcal{C})$ is uniquely defined. In particular, it is
    independent of the (finite) decomposition of $\mathcal{C}$ into convex
    cones.
  \item If $\max_{i \in [N]}\dim \mathcal{C}_{i} \leq k$, then
    $\Delta(\mathcal{C})$ is a union of no more than $N$ at-most $k$-dimensional
    linear subspaces.
  \item The set $\Delta(\mathcal{C})$ satisfies
    $\mathcal{C} \subseteq \Delta(\mathcal{C}) \subseteq
    \mathcal{C}-\mathcal{C}$.
  \item There are choices of $\mathcal{C}$ for which
    $\mathcal{C} \subsetneq \Delta(\mathcal{C})$ (for instance, refer to the
    example at the end of this section).
\end{enumerate}
\vspace{-12pt}
\end{remark}

\begin{proof}[Proof of uniqueness of $\Delta(\mathcal{C})$]
  We begin with the following lemma.
  \begin{lemma}
    \label{lem:span-component}
    Let $N$, $M$ be positive integers, and let
    $\mathcal{C}_1, \mathcal{C}_2, \ldots, \mathcal{C}_N$ be convex cones in
    $\reals^M$ such that their union $\mathcal{C}$ is also convex. Then, there
    exists $i \in [N]$ such that $\Span(\mathcal{C}) = \Span(\mathcal{C}_i)$.
  \end{lemma}
  \begin{proof}[Proof of~\autoref{lem:span-component}]
    First, assume $\Span(\mathcal{C}) = \reals^{M}$. Let $\mu$ be the Gaussian
    measure, so that for a (measurable) set $\mathcal{A} \subseteq \reals^M$
    $\mu(\mathcal{A}) = \pr \{z \in \mathcal{A}\}$, where $z \subset \reals^M$
    has independent standard normal entries. One can show that for any convex
    cone $\mathcal{B}$, $\mu(\mathcal{B}) > 0 $ if and only if
    $\Span(\mathcal{B}) = \reals^M$. Then, since $\mathcal{C}$ is convex,
    \begin{align*}
      0 %
      < \mu(\mathcal{C}) %
      = \mu\left(\bigcup_{i=1}^N \mathcal{C}_i\right) %
      \leq \sum_{i=1}^N \mu(\mathcal{C}_i).
    \end{align*}
    Thus, there exists $i \in [N]$ such that $\mu(\mathcal{C}_i) > 0$, whence
    $\Span(\mathcal{C}_i) = \reals^M$. This proves the lemma when
    $\Span(\mathcal{C}) = \reals^M$.
    
    For the general case, where $\Span(\mathcal{C}) = \mathcal{S}$ for some
    subspace $\mathcal{S}$, replace the Gaussian measure on $\reals^M$ by the
    Gaussian measure restricted to that subspace, $\mu_{\mathcal{S}}$, the
    unique measure satisfying
    $\mu_{\mathcal{S}}(\mathcal{A}) = P( \Pi_{\mathcal{S}} (z) \in
    \mathcal{A})$. As above, for any convex cone $\mathcal{B}$, clearly
    $\mu_\mathcal{S}(\mathcal{B}) > 0$ if and only if
    $\Span(\mathcal{B}) = \mathcal{S}$ and the proof proceeds as above.
  \end{proof}
  
  We now proceed with the proof of uniqueness. Given two collections of convex
  cones
  $\mathcal{F}=\{\mathcal{F}_i\}_{i \in [N]}, \mathcal{D} = \{\mathcal{D}_i\}_{i
    \in [N]}$ such that
  $\mathcal{C} = \bigcup_{i \in [N]} \mathcal{F}_{i} = \bigcup_{i \in [N]}
  \mathcal{D}_{i}$, let
  $\Delta_{\mathcal{F}}(C) := \bigcup_{i \in [N]} \Span(\mathcal{F}_{i})$ and
  define $\Delta_{\mathcal{D}}(C)$ similarly. It is sufficient to show that
  $\Delta_{\mathcal{F}}(\mathcal{C}) = \Delta_{\mathcal{D}}(\mathcal{C})$. Note
  that, without loss of generality, we assume that $\mathcal{F}$ and
  $\mathcal{D}$ have the same number of elements since we can take some cones as
  the empty set.

  Fix $i \in [N]$. Note for all
  $j \in [N],\ \mathcal{F}_{i} \cap \mathcal{D}_{j}$ is a convex cone. Then, we
  have
  \begin{align*}
    \mathcal{F}_i %
    = \mathcal{F}_i\cap \mathcal{C} %
    = \mathcal{F}_i \cap \bigcup_{j \in [N]} \mathcal{D}_{j} %
    = \bigcup_{j \in [N]} (\mathcal{F}_i \cap \mathcal{D}_{j}).
  \end{align*}
  It then follows from~\autoref{lem:span-component}, that there is a $j^*$ such
  that $\Span(\mathcal{F}_{i}) = \Span(\mathcal{F}_{i} \cap \mathcal{D}_{j^*})$.
  Further,
  \begin{align*}
    \Span(\mathcal{F}_{i} \cap \mathcal{D}_{j^*}) %
    \subseteq \Span(\mathcal{D}_{j^*}) %
    \subseteq \Delta(\mathcal{D}).
  \end{align*}
  Since this is true for every $i$,
  $\Delta_{\mathcal{F}}(\mathcal{C}) \subseteq
  \Delta_{\mathcal{D}}(\mathcal{C})$. By symmetry,
  $\Delta_{\mathcal{D}}(\mathcal{C}) \subseteq
  \Delta_{\mathcal{F}}(\mathcal{C})$, so
  $\Delta_{\mathcal{D}}(\mathcal{C}) = \Delta_{\mathcal{F}}(\mathcal{C})$.
\end{proof}

We conclude this section by briefly illuminating the effect of the piecewise
linear expansion $\Delta$ using a simple example.
\begin{myexample}
  Define the $(3,2,3)$-generative network:
  \begin{align*}
    G(z) := \sigma \left(W^{(1)} z\right), %
    \qquad %
    W^{(1)} := \bmat{rrr}{%
    1 & 0 & 0 \\
    -1 & 0 & 0 \\
    0 & 0 & 0}. 
  \end{align*}
  Note the second weight matrix for $G$ is the identity matrix. Observe that
  $\range(W^{(1)}) = \{(x,-x,0) : x \in \reals\}$, so that
  \begin{align*}
    \range(G) = \{\gamma e_1 : \gamma \geq 0\} \cup \{\gamma e_2 : \gamma \geq 0\}.
  \end{align*}
  Consequently, it is straightforward to show that
  \begin{align*}
    \Delta(\range(G)) &= \Span \{e_1\} \cup \Span \{e_2\}
    \\
    \range(G) - \range(G) &= \{(x,y,0) : xy \leq 0\}
    \\
    \Delta(\range(G) - \range(G)) &= \Span\{e_1, e_2\}.
  \end{align*}
  These sets clearly satisfy the inclusion chain
  $\range(G) \subset \Delta(\range(G)) \subset \range(G)-\range(G) \subset
  \Delta(\range(G) - \range(G))$. Moreover, note that $\Delta(\range(G))$ is a
  particular subset of $1$-sparse vectors in $\reals^3$ while
  $\Delta(\range(G) - \range(G))$ is a particular subset of $2$-sparse vectors
  in $\reals^3$.
\end{myexample}

\section{Controlling approximation error}
\label{sec:control-approx-error}

Suppose that $A \in \complex^{\tilde m \times n}$ is an
$(m, \theta, U)$-subsampled isometry with associated unitary matrix
$U \in \complex^{n \times n}$. Recall that $U^{*} U = I$ where $I$ is the
identity matrix, and that $\theta_{i} \iid \DistBer\left( \frac{m}{n}
\right)$. Let $x \in \complex^{n}$ be a fixed vector such that $\|x\|_{2} =
1$. We would like to estimate the tail
\begin{align*}
  \pr\left\{ \|Ax\|_{2}^{2} \geq t \right\}
\end{align*}
for $t \geq 1$, so as to imply a high probability bound of the form
$\|Ax\|_2^2 \leq t\|x\|_2^2$ for any fixed $x \in \complex^n$. For technical
reasons yet to be clarified, we also assume $t < n/m$.

\subsection{Reformulation}

Observe that $A^*A = \frac{n}{m} U^*\diag(\theta)U $, implying
\begin{align*}
  \|Ax\|_2^2 %
  = x^*A^*Ax %
  = \frac{n}{m}x^*U \diag(\theta) Ux %
  = \frac{n}{m}\sum_{i=1}^n \theta_i |z_i|^2,
\end{align*}
where $z := Ux$ and $\|z\|_2 = \|Ux\|_2 = 1$. Hence, the problem can be recast
as bounding the following tail:
\begin{align*}
  \pr\left\{\frac{n}{m}\sum_{i=1}^n \theta_i |z_i|^2 \geq t\right\},
\end{align*}
for $t > 0$ and where $z \in \complex^n$ is a fixed vector with $\|z\|_2 = 1$.

\subsection{Special cases}
\subsubsection{Spike}

Assume that $z := e_{1}$ is the first standard basis vector. Then
\begin{align*}
  \pr\left\{ \frac{n}{m}\sum_{i=1}^n \theta_i |z_i|^2 \geq t \right\} %
  = \pr\left\{ \theta_1 \geq t \frac{m}{n} \right\} %
  = \frac{m}{n} \cdot \1\{t \leq n/m\}
\end{align*}
for any $t > 0$. In particular, one cannot expect good concentration of
``spiky'' vectors.

\subsubsection{Flat vector}
\label{sec:flat-vector}

Assume that $z := \mathbf{1} / \sqrt n \in \complex^{n}$ is the scaled all-ones
vector. In this case, we have
\begin{align*}
  \pr\left\{ \frac{n}{m} \sum_{i=1}^{n} \theta_{i} |z_{i}|^{2} \geq t \right\} %
  = \pr\left\{ \sum_{i=1}^{n} \theta_{i} \geq tm \right\}
\end{align*}
which is the tail of a $\DistBinom(n, p)$ random variable where $p := m/n$.

We use the Cram\'er-Chernoff bound (see~\autoref{sec:cram-chern-bound}
and~\cite[Ch.~2.2]{boucheron2013concentration}). For a centered binomial random
variable $Z := Y - np$ where $Y \sim \DistBinom(n, p)$, the Cram\'er transform
of $Z$ is given by, for all $0 < t < n ( 1- p)$,
\begin{align*}
  \psi_{Z}^{*}(t) := n h_{p} \left(  \frac{t}{n} + p \right), 
\end{align*}
where
\begin{align*}
  h_{p}(a) := (1 - a) \log \frac{1-a}{1-p} + a \log \frac{a}{p} 
\end{align*}
is the Kullback-Leibler divergence $\kldiv{P_{a}}{P_{p}}$ between Bernoulli
distributions with parameters $a$ and $p$. Thus, as presented
in~\cite[Ch.~2.2]{boucheron2013concentration},
\begin{align}
  \label{eq:cramer-chernoff}
  \pr\left\{ Y \geq tm \right\} %
  = \pr\left\{ Z \geq (t - 1)m \right\} %
  \leq \exp(-n h_{p}(tp))
\end{align}
where $p = m/ n$. Note that $(t-1)m < n (1-p)$ since $t < n/m$ as assumed at the
very beginning. Now, assuming $t > 1$, we compute
\begin{align*}
  h_{p}(tp) %
  &= (1 - tp) \log \frac{1-tp}{1-p} + tp \log t
  \\
  & = tp \log t - (1 - tp) \log \left( 1+ \frac{(t-1)p}{1 - tp} \right) %
  \\
  & \geq tp \log t - (t-1)p
\end{align*}
using that $\log (1 + x) \leq x$. Simplifying the above expression gives, using
$p = m/n$,
\begin{align}
  \label{eq:kl-lower-bound}
  h_{p}(tp) \geq \frac{m}{n} \left( t \log t -t +1 \right).
\end{align}
We remark, as an aside, that this lower bound is
nonnegative. Combining~\eqref{eq:cramer-chernoff} and~\eqref{eq:kl-lower-bound}
gives
\begin{align*}
  \pr\left\{ \frac{n}{m} \sum_{i=1}^{n} \theta_{i} |z_{i}|^{2} \geq t \right\} %
  \leq \exp \left( - m \left( t \log t -t +1 \right)  \right). 
\end{align*}
The above inequality gives an explicit quantification of how concentration
occurs in the case of a perfectly flat vector.

\subsubsection{``Flattish'' vector}

The inequality extends as follows.

\begin{proposition}
  \label{prop:cramer-chernoff-concentration}
  Let $\xi \in \complex^{n}$ be a fixed vector satisfying $\|\xi\|_{2} = 1$ and
  let $A \in \complex^{\tilde m \times n}$ be a subsampled isometry associated
  to a unitary matrix $U \in \complex^{n \times n}$. Define
  $R := n\cdot \|\xi\|_{U}^{2}$. Then, for all $1 < t < \frac{n}{m}$,
  \begin{align*}
    \pr\left\{ \left| \|A\xi\|_{2}^{2} - 1 \right| \geq t \right\} %
    \leq \exp \left(  - m \left( \frac{t}{R} \log \frac{t}{R} - \frac{t}{R} + 1 \right) \right).
  \end{align*}
\end{proposition}

The proof of the result follows readily from the discussion for a flat vector
in~\autoref{sec:flat-vector}. Assume instead that $\|z\|_{2} = 1$ and
$\|z\|_{\infty}^{2} \leq R/n$. It is straightforward to show that, for
$t < n /m$,
\begin{align*}
  &\pr\left\{ \frac{n}{m} \sum_{i=1}^{n} \theta_{i} |z_{i}|^{2} \geq t \right\} %
  \\
  & \leq \pr\left\{ \frac{n}{m} \|z\|_{\infty}^{2} \sum_{i=1}^{n} \theta_{i} \geq t \right\} %
  \\
  & = \pr\left\{ \sum_{i=1}^{n}\theta_{i} \geq \frac{tm}{R} \right\} %
  \\
  & \leq \exp \left(  - m \left( \frac{t}{R} \log \frac{t}{R} - \frac{t}{R} + 1 \right) \right).
\end{align*}

\section{Epilogue}
\label{sec:epilogue}

Using a different line of analysis in the vein of~\cite{rudelson2008sparse}
and~\cite[Ch.~4]{dirksen2015tail}, one may proceed under the $\alpha$-coherence
assumption to obtain another set of results that establish restricted isometry
for generative compressed sensing with subsampled isometries, as well as sample
complexity and recovery bounds. Specifically, this line of analysis proceeds via
generic chaining and a careful application of Dudley's inequality, still
leveraging the idea of $\alpha$-coherence. Though this avenue permits one to
weaken the notion of coherence slightly, its proof is significantly more
involved and yields results that are essentially equivalent to the current
work's up to constants. In particular, we conjecture that moving beyond the
$k^{2}$ ``bottleneck'' that we have commented on above
(see~\autoref{sec:conclusion}) is likely to require tools beyond the notion of
$\alpha$-coherence.


\end{document}